\documentclass[12pt,reqno]{amsart}
\usepackage{amsmath,amsfonts,amssymb,bbm}
\usepackage[margin = 1in]{geometry}
\usepackage{booktabs} 

\usepackage[utf8]{inputenc}
\usepackage{ifthen}
\usepackage[boxed, linesnumbered]{algorithm2e}
\usepackage{comment}
\usepackage{graphicx, caption}
\usepackage{pgf}
\usepackage{tikz}
\usetikzlibrary{arrows,automata}
\usepackage{pgfplots}\pgfplotsset{compat=1.9}
\usepgfplotslibrary{fillbetween}

\usepackage{amsthm}
\usepackage{enumerate}
\usepackage{comment}

\usepackage{caption}
\usepackage{subcaption}

\usepackage[colorlinks=true,linkcolor=blue]{hyperref}

\newtheorem{theorem}{Theorem}[section]
\newtheorem{corollary}[theorem]{Corollary}
\newtheorem{lemma}[theorem]{Lemma}

\newtheorem{definition}[theorem]{Definition}
\newtheorem{remark}[theorem]{Remark}
\newtheorem{proposition}[theorem]{Proposition}

\usepackage[utf8]{inputenc} % allow utf-8 input
\usepackage[T1]{fontenc}    % use 8-bit T1 fonts
\usepackage{lmodern}
\usepackage{hyperref}       % hyperlinks
\usepackage{url}            % simple URL typesetting
\usepackage{booktabs}       % professional-quality tables
\usepackage{nicefrac}       % compact symbols for 1/2, etc.
\usepackage{microtype}      % microtypography

\newcommand{\eps}{\mbox{$\varepsilon$}}

\newcommand{\fp}{\overline{y}}
\newcommand{\olx}{\overline{x}}

\def\ol{\overline}
\def\fabs{f_{ab\sigma}}

\def\R{\ensuremath{\mathbb{R}}}
\def\N{\ensuremath{\mathbf{N}}}

\def\ox{\overline{x}}
\def\oxa{\frac{d\ox}{da}}

\SetArgSty{textrm}  % for algorithm2e
\SetAlFnt{\small}
\SetAlCapFnt{\small}
\SetAlCapNameFnt{\small}
\SetAlCapHSkip{0pt}
\IncMargin{-\parindent}

\title{Unpredictable dynamics in congestion games:\\  memory loss can  prevent chaos}

\author[J. Bielawski]{Jakub Bielawski}
\address[J. Bielawski]{Department of Mathematics, Cracow University
  of Economics, Ra\-ko\-wicka~27, 31-510 Krak\'ow, Poland}
\email{jakub.bielawski@uek.krakow.pl}

\author[T. Chotibut]{Thiparat Chotibut}
\address[T. Chotibut]{
Chula Intelligent and Complex Systems, Department of Physics, Faculty of Science, Chulalongkorn University, Bangkok 10330, Thailand.}
\email{Thiparat.C@chula.ac.th, thiparatc@gmail.com}

\author[F. Falniowski]{Fryderyk Falniowski}
\address[F. Falniowski]{Department of Mathematics, Cracow University
  of Economics, Ra\-ko\-wicka~27, 31-510 Krak\'ow, Poland}
\email{falniowf@uek.krakow.pl}

\author[M. Misiurewicz]{Micha{\l} Misiurewicz}
\address[M. Misiurewicz]{Department of Mathematical Sciences, Indiana
  University-Purdue University Indianapolis, 402 N. Blackford
  Street, Indianapolis, IN 46202, USA}
\email{mmisiure@math.iupui.edu}

\author[G. Piliouras]{Georgios Piliouras}
\address[G. Piliouras]{Engineering Systems and Design, Singapore
  University of Technology and Design, 8 Somapah Road, Singapore 487372}
\email{georgios@sutd.edu.sg}

\begin{document}

\maketitle

%\markleft{BIELAWSKI, CHOTIBUT, FALNIOWSKI, MISIUREWICZ, PILIOURAS}

\begin{abstract}
We study the dynamics of simple congestion games with two resources where a continuum of agents behaves according to a version of Experience-Weighted Attraction (EWA) algorithm. The dynamics is 
characterized by two parameters: the (population) intensity of choice $a>0$ capturing the economic rationality of the total population of agents and a discount factor $\sigma\in [0,1]$ capturing a type of memory loss where past outcomes matter exponentially less than the recent ones. Finally, our system adds a third parameter $b \in (0,1)$, which captures the asymmetry of the cost functions of the two resources. It is the proportion of the agents using the first resource at Nash equilibrium, with $b=1/2$ capturing a symmetric network.

Within this simple framework, we show a plethora of bifurcation phenomena where behavioral dynamics destabilize from global convergence to equilibrium, to limit cycles or even (formally proven) chaos as a function of the parameters $a$, $b$ 
and $\sigma$. Specifically, we show that for any discount factor $\sigma$ the system will be destabilized for a sufficiently large intensity of choice $a$. Although for discount factor $\sigma=0$ almost always (i.e., $b \neq 1/2$) the system will become chaotic, as $\sigma$ increases the chaotic regime will give place to the attracting periodic orbit of period 2. Therefore, memory loss can simplify game dynamics and make the system predictable. We complement our theoretical analysis with simulations and several bifurcation diagrams that showcase the unyielding complexity of the population dynamics (e.g., attracting periodic orbits of different lengths) even in the simplest possible potential games.

\end{abstract}

\section{Introduction}
\label{s:intro}

Congestion games~\cite{rosenthal73} are arguably amongst the most well-studied classes of games in game theory. As their
name indicates, they capture multi-agent settings where the costs of each agent depends on the
resources she chooses and how congested each of them is (e.g., traffic routing, common resources). Congestion games are well known to be isomorphic
to potential games~\cite{monderer1996potential}, i.e. games where the incentives of all agents are perfectly aligned with each other by being equivalent to optimizing a single potential function. Furthermore, non-atomic (also known as population) potential/congestion games are  even more regular game settings as under a minimal natural assumption on the cost functions of resources %(i.e., non-negative, continuous, non-decreasing)
it is known that they admit an essentially unique equilibrium flow which coincides with the global minimum of a strictly convex potential function~\cite{NisaRougTardVazi07}. Given the above, (population) potential/congestion games are typically thought as the paragon of game theoretic stability with numerous evolutionary dynamics provably converging in them via Lyapunov arguments where the potential function is strictly decreasing with time (e.g., \cite{chen2016generalized,cohen2017learning,Even-Dar:2005:FCS:1070432.1070541,Fischer:2006:FCW:1132516.1132608,Fotakis08,Kleinberg09multiplicativeupdates,kleinberg2011load,krichene2015online,mertikopoulos2018riemannian,palaiopanos2017multiplicative,Sandholm10}). In fact, these games are treated as testing grounds for novel game dynamics, where convergence is treated more like  % as little more than a foregone conclusion as well as 
a minimal desideratum with interesting, novel technicalities rather than a major insight into the setting itself.

Despite the ubiquitous nature of these positive results, at a closer look, a common driving ``regularity'' assumption emerges at their core. The behavioral dynamics have to be in a sense ``smooth'' enough to act as a gradient-like system for the common potential function. This type of regularities typically follow automatically in the case of continuous-time dynamics (e.g., \cite{Sandholm10}) whereas in the case of discrete-time dynamics they can be enforced by appropriate upper bounds on the step-size/learning rates. These bounds decrease as we increase the Lipschitz constant of the gradient of the potential with steeper potentials resulting in more restrictive bounds on the intensity of choice of the agents. In the case of non-atomic congestion games, as we increase the total population size (i.e. total load/congestion) these bounds converge to zero. Of course, the mathematical necessity of such regularity conditions is abundantly clear as even with gradient descent on a strictly convex function the step-size has to be controlled as the function becomes steeper to avoid overshooting effects. 
 What is less clear is how well do these mathematically driven constraints agree with our best known understanding of how people actually behave and adapt when facing such strategic considerations in practice. Which parameters are important in practice and how fine-tuned can we expect them to be in large population settings?

 The question of how people learn to adapt their strategies in real-world games is the object of study of behavioral game theory~\cite{camerer1999experience,camerer2011behavioral,ho2007self}. 
Experience-Weighted Attraction (EWA) is a canonical learning model in behavioral game theory and although in its full generality contains too many free parameters recent work has focused on a stripped down version~\cite{GallaFarmer_PNAS2013}
which allows only two free parameters, the intensity of choice (akin to the exponent of a logit choice function~(e.g., \cite{alos2010logit})) and a memory loss parameter, which is akin to the rate of exponential discounting over past payoffs. 
We adopt this model where, roughly speaking, agents perform logit best-responses to an estimate of the historical performance of each action where the effect of past payoffs/costs decays exponentially fast.  Importantly, experimental work in the area suggests that large intensity of choice (sometimes in double digits) are common, pointing out at the possibility of a conflict between standard mathematically driven assumptions and experimentally tested behavioral regularities. Similarly, the discounting rate is shown to be reliably positive; however, its actual value can vary widely from game to game (see e.g., Table 4 in \cite{ho2007self}). How does the interplay between intensity of choice and memory loss affect the convergence results in simple congestion games? At what points of the parameter space do the dynamics destabilize and when they do what sort of behavior do they give rise to (limit cycles, chaos, etc.)?  Finally, how does the nature of the congestion game (e.g. symmetry of costs) affect its stability? In a prior conference proceeding ~\cite{CFMP2019}, we have provided some preliminary answers to these questions for the  special case where the agents have no memory loss, whereas now we provide a more thorough understanding of the complex phenomena emerging at different levels of  memory loss.

%Prior to our paper, the only known results in this area, reflected the special case where the agents have no memory loss. In this case, our behavioural learning model contains as a special case, a ubiquitous optimisation algorithm known as Multiplicative Weights Updates (MWU~\cite{}) also known as Hedge~\cite{} (see also related work for other closely related algorithms). In this special case, that instability can emerge in simple non-atomic congestion games as long as the either the effective learning rate of the system is large enough (either the total demand becomes too large, or the intensity of choice of the individual agents? becomes too large).

{\bf Our results.} Due to a memory loss parameter, the (interior) fixed point of these different learning dynamics corresponds to a unique perturbed equilibrium.  As the intensity of choice increases perturbed equilibrium approaches Nash equilibrium~(Proposition~\ref{qreolx}) but at the cost of losing stability.
 As long as the equilibrium is locally attracting then it is globally attracting (Theorem~\ref{afp}).
 There is a lower bound on the learning rate above which the equilibrium is repelling. This threshold value is decreasing with the discount factor (Proposition~\ref{QREstability}).
 If there is memory loss, then there exists a threshold value on the effective learning rate of the system  such that if the intensity of choice crosses this threshold, then the game is unstable no matter what the costs are. In contrast, if there is no memory loss, for any intensity of choice one can find congestion games such that the dynamics is stable (Proposition~\ref{abthreshold}).

Next, we study exactly how unpredictable the system will become when stability is lost. We start with the simplest case where the dynamics is memoryless (i.e. only last period costs matter). In this case,  
there exists an intensity of choice threshold such that below it the dynamics are globally stable whereas above it almost all but a countably large set of initial conditions converge to a limit cycle of length two (Proposition~\ref{dynsigma0}).  In the case of systems with memory, 
as long as the underlying congestion game is not very asymmetric, then under large enough intensity of choice the system will inevitably converge to an attracting periodic orbit of period two (Theorem~\ref{thm2period}).
Thus, although the system won't converge, the behavior will be predictable. On the other hand, if the resources have significantly different costs resulting in a Nash equilibrium flow $b$ far from the symmetric $50\%-50\%$ split (i.e. the value $b$ is far from $0.5$) between the two resources, then there is a threshold value for the intensity of choice above which the dynamics are Li-Yorke chaotic (Theorem~\ref{chaos}). Putting everything together we derive that for large enough intensity of choice, as long as $b \in (\frac{1-\sigma}{2-\sigma}, \frac{1}{2-\sigma})$ almost all trajectories are attracted to the periodic orbit of period two, whereas outside this interval we observe chaos 
	(see Figure~\ref{fig: cp2}). 

\begin{figure}[t]
\centering
\includegraphics[width=0.65\textwidth]{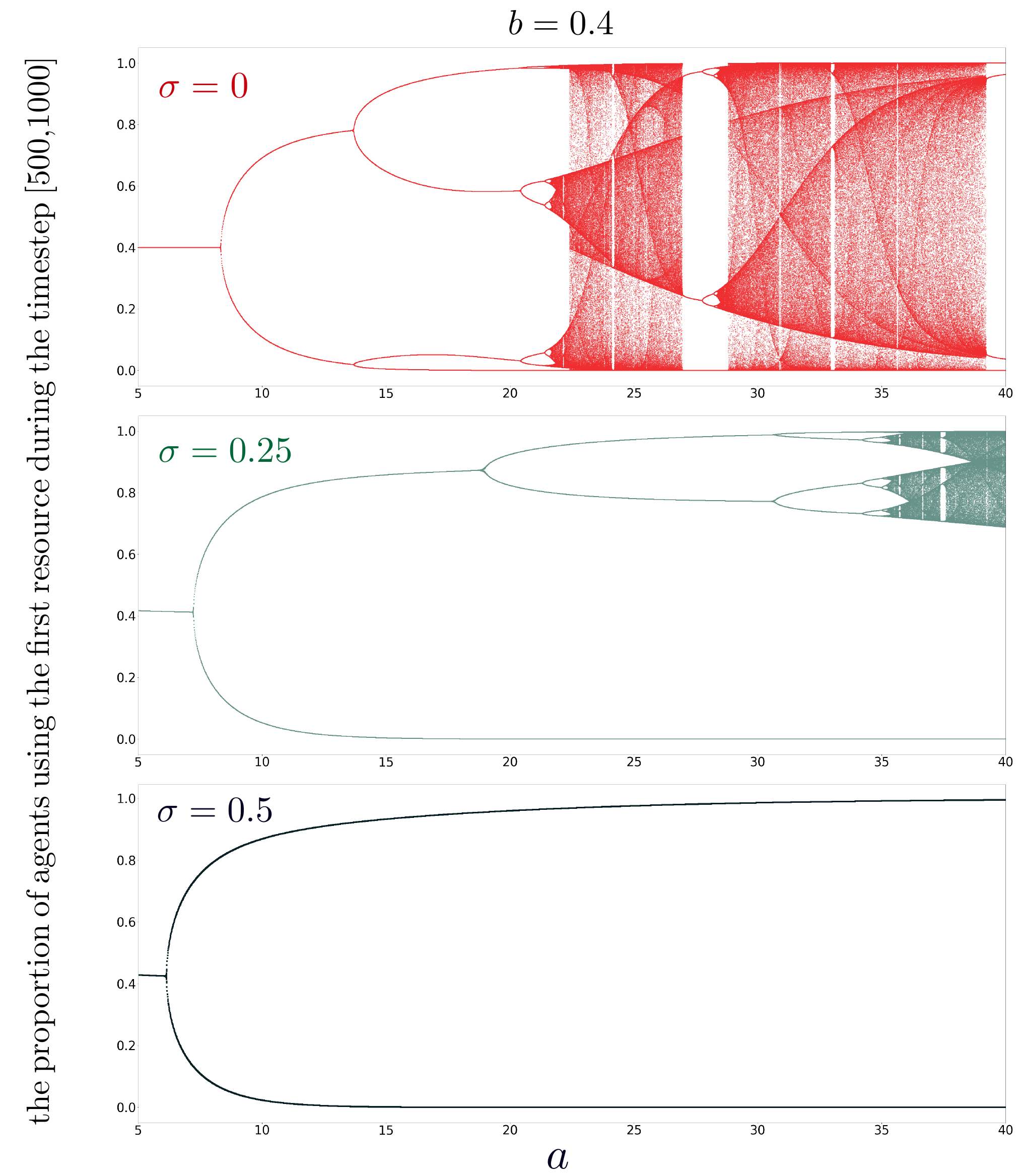} %0.95  0.68
\begin{minipage}{0.9\textwidth}
\caption{Bifurcation diagrams for increasing values of the discount factor $\sigma$ when the Nash equilibrium is fixed at $b=0.4$. Without any discounting ($\sigma = 0$), the dynamics is described by the Multiplicative Weight Update which leads to the period-doubling bifurcation route to chaos as one increases the intensity of choice $a$~ \cite{CFMP2019}. 
%Moreover, for each $a$, the time average of the dynamics converges {\it exactly} to the Nash equilibrium $b$ \cite{CFMP2019}, shown in the green horizontal line. 
However, as the discount factor increases to $\sigma = 0.25$, chaos starts at a much larger intensity of choice. 
%The time average also no longer converges to $b$ %{\color{red}[[refer to the equation of the time average convergence of other quantities.]]}
As the discount factor increases to $\sigma = 0.5$, increasing the intensity of choice can at most lead to the period-2 dynamics. Thus, a larger discount factor tends to make the dynamics more stable and predictable.}
\label{fig: bifurcation_diagrams}
\end{minipage}
\end{figure}

We complement our theoretical understanding with further simulations and numerical experiments.
In Figure~\ref{fig: bifurcation_diagrams} we present
bifurcation diagrams at increasing values of the discount factor/memory loss $\sigma$ in a specific instance of a two-resource congestion game where the Nash equilibrium flow is fixed at $b=0.4$. Without any discounting ($\sigma = 0$), the dynamics
%is described by Multiplicative Weight Update (MWU) which
leads to the period-doubling bifurcation route to chaos as one increases the intensity of choice $a$~\cite{CFMP2019}. 
%Moreover, for each $a$, the time average of the dynamics converges {\it exactly} to the Nash equilibrium $b$ \cite{CFMP2019}, shown in the green horizontal line. However, 
As the discount factor increases to $\sigma = 0.25$, chaos starts at a much larger intensity of choice. 
%The time average also no longer converges to $b$ %{\color{red}[[refer to the equation of the time average convergence of other quantities.]]}
Finally, as the discount factor increases to $\sigma = 0.5$, increasing the intensity of choice can at most result in the existence of a period-2 limit cycle. Thus, a larger discount factor tends to stabilize the dynamics.
%; however, the time-average convergence to the Nash equilibrium property is lost. 
In Figure \ref{fig: cobwebb_b0p7} we show how
increasing the discount factor/memory loss $\sigma$ can reduce chaotic, unpredictable dynamics into periodic, predictable ones by examining both the day-to-day behavior as well as the time evolution on the potential function of the game. 

\begin{figure}
  \centering
   \includegraphics[width = 0.72\textwidth]{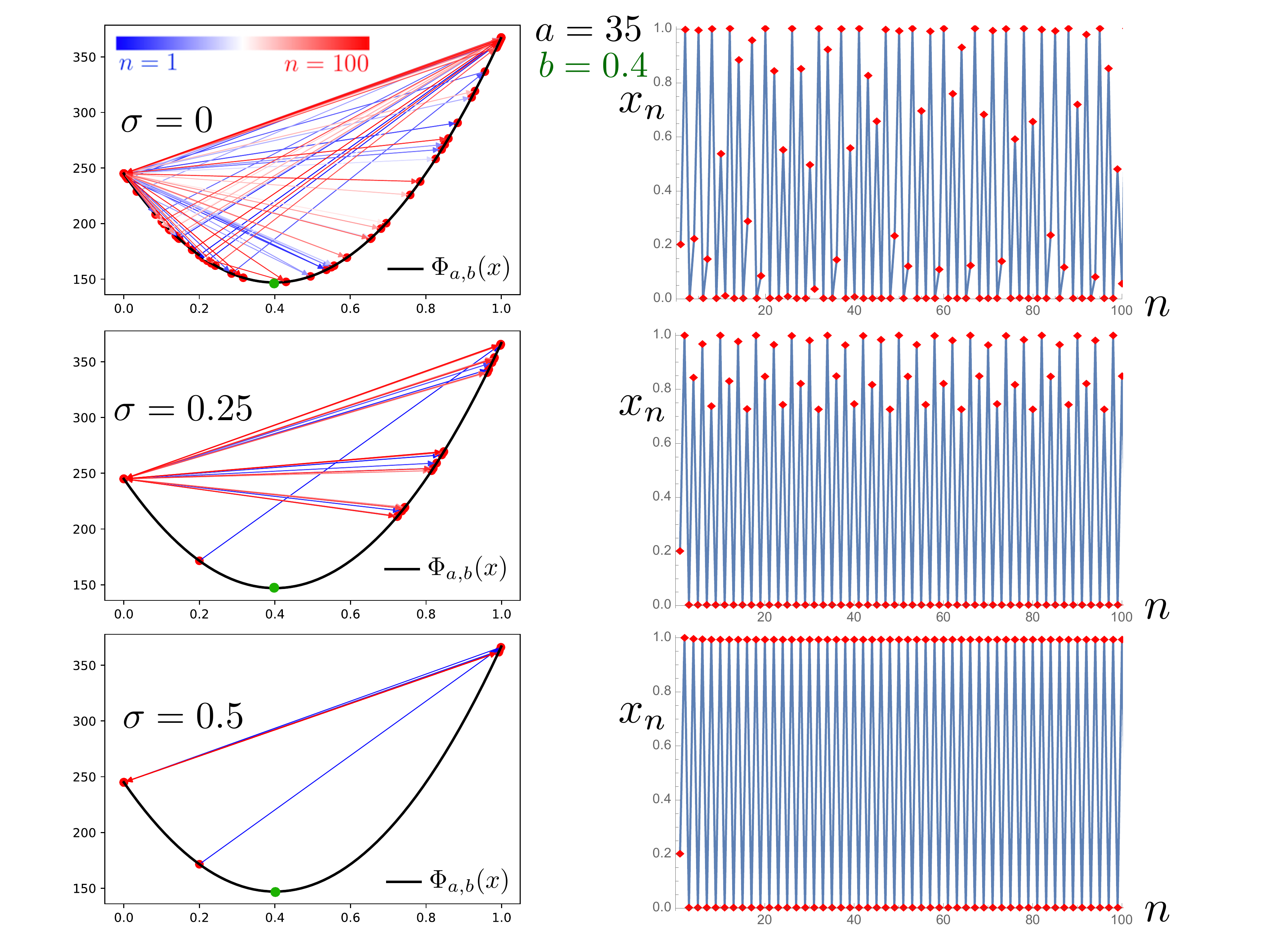} 
   \begin{minipage}{0.9\textwidth}
   \caption{\small A larger discount factor $\sigma$ help stabilize the dynamics. These figures are visualizations of the dynamics that lead to the bifurcation diagrams of Figure \ref{fig: bifurcation_diagrams} at a fixed $a=35$. The initial state is set to $x_0 = 0.2.$ The left column shows the dynamics in the convex potential (cost) landscape in our congestion game $\Phi_{a,b}(x) \equiv \frac{N^2}{2}\left( \alpha x^2 + \beta (1-x)^2\right) = \frac{a^2}{2}\left( (1-b)x^2 + b(1-x)^2\right)$, whose Nash equilibrium $b=0.4$ is the potential minimum. The right column shows the dynamics of $x_n$. For this large $a=35$, larger discount factors tend to stabilize the chaotic dynamics to at most a period 2 instability. The top row with $\sigma =0$ corresponds to the Multiplicative Weights Update dynamics  \cite{CFMP2019}, whose time-average of $x_n$ is exactly the Nash equilibrium $b$ \cite{CFMP2019}. }  
      \label{fig: cobwebb_b0p7}
   \end{minipage}
\end{figure}

\section{Related literature}
\label{s:related}

%\begin{itemize}
%\item learning \cite{benaim2009learning}
%\item learning \cite{leslie2005individual,CGM,benaim2009learning,leonardos2021exploration,cominetti2010payoff}
% \item perturbed dynamics: cited \cite{hofbauer2002global,hofbauer2007evolution,fudenberg1993learning}, to cite \cite{nicole2017stochastic}
%\item general \cite{sandholm2010pairwise,Sandholm10,hart2000simple,hart2003uncoupled}
%\end{itemize}

Concepts discussed in this article arise in different contexts and were studied independently by many economists and computer scientists.
%One of the most influential theoretical approach is perturbed best-response dynamics.
The problem of introducing discounting of the past costs (or random mistakes) is an important issue both from theoretical and experimental economics perspective. Two concepts closely related to the model discussed in this article are  Experience-Weighted Attraction (experimental one, usually discrete time)  and perturbed best response dynamics (theoretical one, usually continuous time).

 {\bf Experience-Weighted Attraction (EWA).}
 EWA is arguably one of the most influential learning models in behavioral game theory~\cite{camerer1999experience,camerer2011behavioral,camerer2002sophisticated}. 
% One of the criticism of the EWA model is that it contains many free parameters and thus recent work has focused on producing more stripped down versions of EWA, e.g.~\cite{ho2007self}.
 In its complete form EWA contains many free parameters, thus recent work has focused on its stripped down version~\cite{GallaFarmer_PNAS2013}
which allows only two free parameters, the intensity of choice (akin to the exponent of a logit choice function~(e.g., \cite{alos2010logit})) and a memory-loss parameter, which is akin to the rate of exponential discounting over past payoffs.
 The version of EWA that we study in this paper follows from \cite{GallaFarmer_PNAS2013}, which includes these two free parameters. %, the intensity of choice and a memory-loss parameter.
 In \cite{Farmer2019}, it is shown that best reply cycles can predict non-convergence  of six  well-known  learning  algorithms in games with random payoffs where one of the algorithms considered is EWA. Other dynamics include replicator dynamics, reinforcement learning, fictitious play and $k$-level EWA, showing that often there exist similarities between their behavior at least in small, randomly chosen games.

{\bf Perturbed dynamics.} There is rich literature on perturbed best response dynamics. It was initiated by Fudenberg and Kreps \cite{fudenberg1993learning} who introduced stochastic fictitious play, where each agent's payoffs are perturbed in each period by random shocks a la Harsanyi \cite{Harsanyi1973}. As a consequence, each player's anticipated behavior in each period is a genuine mixed strategy. Hofbauer and Sandholm \cite{hofbauer2002global} showed convergence of stochastic fictitious play in games with an interior ESS, zero-sum games, potential games and supermodular games.\footnote{There is no convergence in the general case (see e.g. \cite{benaim1999mixed})} However, the steady state to which the dynamics converges is not the Nash equilibrium but perturbed equilibrium (called also quantal response equilibrium, see \cite{goeree2020stochastic,mckelvey1995quantal}).
Moreover, Hopkins \cite{hopkins2002two} showed that stochastic fictitious play and perturbed reinforcement learning \cite{erev1998predicting} can be considered as noisy versions of the evolutionary replicator dynamics in two-player games. Thus, stationary points of a perturbed reinforcement learning model will be identical to those of stochastic fictitious play, and if the two models do converge, they will converge to the same state. In addition, if any mixed equilibrium is locally stable (unstable) for all forms of stochastic fictitious play, it is locally stable (unstable) for perturbed reinforcement learning. A number of authors have since generalized dynamic to the class of perturbed best response dynamics \cite{benaim1999mixed,benaim2009learning,hofbauer2002global,hofbauer2007evolution} in the context of large population games with finite strategy sets.\footnote{See Sandholm \cite{Sandholm10} for an extensive review of the literature on finite strategy (continuous) evolutionary game theory.} Recently it was also generalized for continuous games \cite{lahkar2021generalized}.  %{\color{red} stronger exposition of potential games}

%Another perturbed dynamics was considered by Hofbauer and Sandholm \cite{hofbauer2007evolution}. In their model each agent occasionally receives opportunities to update his choice of strategy. When such an opportunity arises, the agent selects a strategy that is currently optimal, but only after payoffs have randomly perturbed. They obtain convergence to approximate Nash equilibrium in stable, potential and supermodular games.

%The logit dynamic \cite{} is the prototypical representative of dynamics obtained by perturbing the best response thorugh the Shannon entropy function.

 {\bf Learning in games.} 
 %Learning in games have received considerable attention over the last 30 years.
Learning procedures can be divided into two broad categories depending on whether they evolve in continuous or discrete time: the former includes the numerous dynamics for learning and evolution (see Sandholm \cite{Sandholm10} and Hadikhanloo et al. \cite{hadikhanloo2021learning} for recent surveys), whereas the latter focuses on learning algorithms (such as fictitious play and its variants) for infinitely iterated games \cite{fudenberg1998theory}.

The EWA algorithm, discussed in this paper, can be seen as a reinforcement learning algorithm where agents score their actions over time based on their observed payoffs and then they choose an approximate/perturbed best response.
Learning algorithms of this kind have been investigated in continuous time by B\"{o}rgers and Sarin \cite{borgers1997learning}, Hopkins \cite{hopkins2002two}, Coucheney et al. \cite{CGM} and many others.
The model closely related to the one discussed in this paper was proposed by Coucheney et al. \cite{CGM}. They derived a class of penalty-regulated game dynamics consisted of replicator-like drift plus a penalty term that keeps agents from approaching the boundary of the state space. These dynamics are equivalent to agents scoring their actions by comparing their exponentially discounted cumulative payoffs over time and using a smooth best response to pick an action. They show global convergence for the continuous case.

Our model is also related to discrete-time model of $Q$-learning \cite{hopkins2002two,leslie2005individual,tuyls2006evolutionary}. From a discrete-time viewpoint Leslie and Collins \cite{leslie2005individual} used a $Q$-learning approach to establish the convergence of the resulting learning algorithm in two-player games under minimal information assumptions, a similar approach was also used by Cominetti et al. \cite{cominetti2010payoff}. 
Finally, the model presented in this article subsumes two well-known dynamics: a discrete-time variant of replicator dynamics  (Multiplicative Weights Update) \cite{freund1999adaptive} and logit best-response \cite{alos2010logit,blume1993statistical}. Moreover, it is a two-parameter version of EWA dynamics which was numerically studied for random zero-sum games by Galla and Farmer \cite{GallaFarmer_PNAS2013}  and Pangallo et al. \cite{Farmer2019}.

%From the learning perspective there a vast literature on this subject. In this context one can look at the perturbed dynamics via reinforcement learning and penalty-regulated dynamics \cite{CGM,MS2016}. To the thorough discussion on continuous time case we refer the reader to \cite{CGM,MS2016} and references therein. } {\color{red} expand}

{\bf Chaos in games.} %In reality, players rarely follow a Nash equilibrium strategy. 
The question one may ask is how complicated the behavior of agents may become even in simple games. The seminal work of \cite{SatoFarmer_PNAS} showed analytically by computing the Lyapunov exponents of the system that even in a simple two-player game of rock-paper-scissor replicator dynamics (the continuous-time variant of MWU) can lead to chaos, rendering the equilibrium strategy inaccessible. Replicator dynamics has recently been shown to be able to produce arbitrarily complex orbits (e.g. Lorenz butterfly dynamics) in simple matrix games~\cite{andrade2021learning}.

 For two-player games with a large number of available strategies (complicated games), \cite{GallaFarmer_PNAS2013} argues that EWA algorithm, exhibits chaotic behavior in a large parameter space. The prevalence of these chaotic dynamics also persists in games with many players, as shown in the follow-up work \cite{GallaFarmer_ScientificReport18}. Careful examinations suggest a complex behavioral landscape in many games (small or large) for which no single theoretical framework currently applies.
 In \cite{becker2007dynamics} and \cite{geller2010microdynamics} a chaotic behavior of Nash maps in game like matching pennies is shown.
  In \cite{VANSTRIEN2008259} and \cite{VANSTRIEN2011262} it is proved that fictitious play learning dynamics for a class of 3x3 games, including the Shapley's game and zero-sum dynamics, possesses rich periodic and chaotic behaviors. \cite{Soda14} showed that replicator dynamics, the continuous-time variant of MWU, is Poincar\'{e} reccurent in zero-sum games. \cite{mertikopoulos2017cycles} generalized this result to  Follow-the-Regularized-Leader (FTRL) algorithms (called also dual averaging). When MWU/FTRL is applied with constant step-size in zero-sum games it becomes unstable~\cite{BaileyEC18} and in fact Lyapunov chaotic~\cite{CP2019}. \cite{2017arXiv170109043P} showed experimentally that EWA leads to limit cycles and high-dimensional chaos
 in two-agent games with negatively correlated payoffs. \cite{cheung2020chaos} established Lyapunov chaos in the case of coordination/potential games for a variant of MWU, known as Optimistic MWU. However, none of the above results implies formal chaos in the sense of Li-Yorke. 
 
 The first formal proof of Li-Yorke chaos was shown for MWU in a single instance of two agent two strategy congestion game in~\cite{palaiopanos2017multiplicative}. This result was generalized for all two-agent two-strategy coordination games in~\cite{Thip18}.
  In arguably the main precursor of our work, \cite{CFMP2019} established Li-Yorke chaos in nonatomic congestion game where agents use MWU. This result was then extended to FTRL with steep regularizers~\cite{BCFKMP21}. The theory of Li-Yorke chaos has since then been applied in other game theoretic settings related to markets~\cite{cheung2021learning}, as well as blockchain protocols~\cite{leonardos2021dynamical}.

\section{Model}\label{s:prelim}

We consider a two-strategy nonatomic \emph{congestion game} (see~\cite{rosenthal73}) with a continuum of agents (players), where all of them apply the \emph{Experience-Weighted Attraction} (EWA) algorithm to update their strategies. Each of the agents controls an infinitesimally small fraction of the flow.
  %We will assume that 
  The total flow of all the agents is equal to $N$. % normalized to $1$. 
  We will denote the fraction of the agents adopting the first strategy at time $n$ as $x_n$. 
  %The second strategy is then chosen by $1-x_n$ fraction of the players. This model physically encapsulates how a large population of commuters selects between the two alternative paths that connect the initial point to the end point.
%When a large fraction of the players adopt the same
%strategy, congestion arises, and the cost of choosing the same strategy increases.

\subsection{Linear congestion games}

%We %primarily 
%focus on linear cost functions. Specifically, 
 The cost of each resource (path, link, route or strategy) here will be assumed proportional to the \emph{load}. By denoting $c_j$ the cost of selecting the
strategy number $j$ (when $x$ fraction of the agents choose the first strategy), if the coefficients of proportionality are
$\alpha,\beta>0$, we obtain
\begin{equation}\label{cost}
%\begin{aligned}
c_1(x) = \alpha N x, \hspace{50pt} c_2(1-x)=\beta N (1-x).
%\end{aligned}
\end{equation}

Without loss of generality we will assume throughout the paper that $\alpha+\beta=1$. Therefore, the values of $\alpha$ and $\beta=1-\alpha$ indicate how different the resource costs  are from each other.\footnote{Our analysis on the emergence of bifurcations, limit cycles and chaos will carry over immediately to the cost functions of the form $\alpha x+\gamma$ for $\alpha, \gamma > 0$.}  As we will see, the only parameter of the game that is important is the value of the equilibrium split, i.e. the fraction of agents using the first strategy at equilibrium. One advantage of this formulation is that the fraction of agents using each strategy at equilibrium is independent of the flow $N$.
%The second advantage is that the Price of Anarchy of these games is exactly $1$, independent of $\alpha, \beta, $ and $N$.  Hence, our model offers a natural benchmark for comparing equilibrium analysis, which suggests optimal social cost, to the time-average social cost arising from non-equilibrium learning dynamics which as we show can be as large as possible.

\subsection{Learning in congestion games with EWA}

Experience-Weighted Attraction (EWA) has been proposed by Camerer and Ho \cite{EWA1} as a stochastic algorithm that subsumes reinforcement learning and belief learning algorithms. This unifying property comes with a consequence of many free parameters. In this paper we focus on a deterministic variant of EWA \cite{GallaFarmer_PNAS2013} which has only two free parameters, an intensity of choice and a memory loss parameter, which can be seen as the rate of exponential discounting of past costs.

We assume that at time $n+1$ the agents know the cost of the strategies at time $n$ (equivalently, the probabilities $(x_n, 1-x_n)$). Since we have a continuum of agents, the realized flow (split) is accurately described by the probabilities $(x_n, 1-x_n)$. There is a parameter $\eps\in(0,1)$, which can be treated as the common learning rate of all agents, such that $\lambda = \log\frac{1}{1-\varepsilon}$ describes the intensity of choice. Then the agents update their choices using EWA algorithm
\begin{align}
\begin{split}\label{mwu}
 x_{n+1} &= \frac{x_n^{1-\sigma} \exp(-\lambda c_1(x_n))}{x_n^{1-\sigma} \exp(-\lambda c_1(x_n))+(1-x_n)^{1-\sigma} \exp(-\lambda c_2(1-x_n))}\\
 &= \frac{x_n^{1-\sigma}}{x_n^{1-\sigma}+(1-x_n)^{1-\sigma} \exp[\lambda (c_1(x_n)-c_2(1-x_n))]}\\
 &= \frac{x_n^{1-\sigma}}{x_n^{1-\sigma}+(1-x_n)^{1-\sigma} \exp(a(x_n-b))},
\end{split}
\end{align}
where $a = N\log\frac{1}{1-\varepsilon} > 0$ is a population intensity of choice, $b = \beta \in (0,1)$ is the equilibrium split, i.e. the fraction of agents using the first strategy at equilibrium and $\sigma \in [0,1]$ is a memory loss parameter.

Note that if $\sigma$ is treated as a discount factor then it describes how individuals value the past: the greater $\sigma$ the less important are previous plays, the more important recent plays.

Equation \eqref{mwu} implies that the dynamics of changes in the behavior of agents (and the system) is governed by the map
\begin{equation}\label{fabd}
f_{ab\sigma}(x) = \frac{x^{1-\sigma}}{x^{1-\sigma}+(1-x)^{1-\sigma}\exp (a(x-b))},
\end{equation}
where $a > 0$, $b \in (0,1)$, $\sigma\in [0,1]$.

%\subsection{Regret, Price of Anarchy and time average social cost}\label{sec: regretPoASC} {\color{red} WHAT DO WE WANT TO SAY ABOUT THAT ?}

\subsection{Equilibrium}

We assume that the population of agents is homogeneous, that is all agents use the same mixed strategy. Hence, adoption of a strategy profile $(x,1-x)$ by agents results in $x$ fraction of the agents choosing the first strategy.

\begin{definition}
A strategy profile $(x,1-x)$ is a {\it Nash equilibrium} if and only if no agent can strictly decrease his/her expected cost by unilaterally deviating to another strategy.
\end{definition}

\begin{definition} \label{pertdef}
We call $(\ol{x},1-\ol{x})$ a perturbed equilibrium, if $\ol{x}$ is a fixed point of the map $f_{ab\sigma}$.
\end{definition}

To show why the term ``perturbed'' is used one can show that it agrees with general concepts of perturbed equilibrium \cite{Sandholm10}. %and quantal response equilibrium \cite{mckelvey1995quantal}.  
Indeed, take an interior fixed point $\ol{x}$ of $\fabs$. By \eqref{mwu} it must satisfy%\footnote{{\color{red} change --- in equation \eqref{fabd} we don't have $\varepsilon$}}
\[ \left\{\begin{array}{l}
  \log(\overline{x}) = (1-\sigma) \log(\overline{x}) + \log(1-\varepsilon) c_1(\overline{x}) - \log Z, \\
  \log(1-\overline{x}) = (1-\sigma) \log(1-\overline{x}) + \log(1-\varepsilon) c_2(1-\overline{x}) - \log Z,
 \end{array}\right.\]
where $Z$ is the denominator of the second term in \eqref{mwu}. From the above system of equations we derive that 
\begin{equation} \label{perteq} \frac{N \sigma}{a} \log\frac{\overline{x}}{1-\overline{x}}  =  c_2(1-\overline{x}) - c_1(\overline{x}). \end{equation}
%Thus, the probability of playing action $i$ is proportional to $e^{-\theta \cdot cost(i)}$ where $\theta=\frac{N\sigma}{a}$. 
Thus, the probability of playing action $i$, denoted by $x_i$, is proportional to $e^{-\frac{a}{N \sigma}c_i(x_i)}$.
Equation \eqref{perteq} shows that when $\sigma=0$ the fixed point $\ol{x}$ is a Nash equilibrium, whereas $\sigma>0$ perturbs the equilibrium state.

%that if $\sigma=0$  %Therefore, the perturbed equilibrium of \eqref{dynamics} corresponds to the notion of quantal response equilibrium.% (when $\sigma \in (0,1)$) or logit equilibrium (when $\sigma = 1$).

\subsection{Attracting orbits and chaos}

%In the remainder of this section 
In this subsection we introduce basic notions of dynamical systems.%: Li-Yorke chaos and attracting/repelling fixed points and periodic orbits. To this end we denote by $f^n$ composition of the map $f$ with itself $n$-times.

\begin{definition}
Let $x$ be a fixed point of a dynamical system $(X,f)$. The fixed point $x$ is called:
\begin{itemize}
\item attracting, if there is an open neighborhood $U \subset X$ of $x$ such that for every $y \in U$ we have $\lim\limits_{n \to \infty} f^n(y) = x$, where $f^n$ is a composition of the map $f$ with itself $n$-times. 
\item repelling, if there is an open neighborhood $U \subset X$ of $x$ such that for every $y \in U$, $y \neq x$ there exists $n \in \N$ such that $f^n(y) \in X\backslash U$.
\end{itemize}
\end{definition}
In this note we study differentiable maps on the unit interval. Then let $x$ be a fixed point, if $|f'(x)| < 1$, then $x$ is attracting,  if $|f'(x)| > 1$, then $x$ is repelling. If  $|f'(x)| =1$ we need more information.

\begin{definition}
Let  $(X,f)$ be a dynamical system. An orbit $\{f^n(x)\}$ is called periodic of period $T$ if $f^{n+T}(x)=f^n(x)$ for any $n\in \mathbf{N}$. The smallest such $T$ is called the period of $x$.
The periodic orbit is called attracting, if $x$ is an attracting fixed point of $(X,f^T)$, and
 repelling, if $x$ is a repelling fixed point of $(X,f^T)$.
\end{definition}

It seems that there is no universally accepted definition of chaotic behavior of a dynamical system. Most definitions of chaos concern one of the following aspects:
\begin{itemize}
\item complex behavior of trajectories, such as Li-Yorke chaos;
\item fast growth of the number of distinguishable orbits of length $n$, such as  positive topological entropy;
%\item existence of absolutely continuous invariant measures;
\item sensitive dependence on initial conditions, such as Devaney or Auslander-Yorke chaos;
\item recurrence properties, such as transitivity or mixing. %and weakly mixing sets.
\end{itemize}

In this article, the first two are crucial. %and survey connections between them.
Also, in the presence of chaos, studying precise single orbit dynamics can be intractable.
%; we study the average behavior of trajectories instead. Thus, it is important to know whether the average converges. This is when ergodic theorems come into play. 

\begin{definition}[Li-Yorke chaos]  \label{LYchaos-def}
Let $(X,f)$ be a dynamical system and $x,y \in X$.
We say that $(x,y)$ is a \emph{Li-Yorke pair} %\footnote{Such pairs appeared in the seminal paper of Li and Yorke \cite{liyorke}.} 
if
\[ \liminf_{n\to\infty} dist (f^n(x),f^n(y))=0,\;\text{and}\;
\limsup_{n\to\infty} dist (f^n(x),f^n(y))>0.\]
A dynamical system $(X,f)$ is \emph{Li-Yorke chaotic} if there is an uncountable set $S\subset X$ (called \emph{scrambled set}) such that every pair $(x,y)$ with $x,y\in S$ and $x\neq y$ is a Li-Yorke pair. 
\end{definition}

Li-Yorke chaos occurs when there is an uncountable set $S \subset X$, such that for every pair of points from $S$ their trajectories are close to each other infinitely many times and they are far from each other infinitely many times.

The origin of the definition of Li-Yorke chaos is in the seminal Li and Yorke's article \cite{liyorke}.
Intuitively orbits of two points from the scrambled set have to gather themselves arbitrarily close and  spring aside infinitely many times but (if $X$ is compact) it cannot happen simultaneously for each pair of points. %That's why we say about {\it scrambled} set.
 Why should a system with this property be chaotic? %That is not completely intuitive, one may discuss every single ingredient of the definition, but it turns out to make sense as it is. 
Obviously the existence of a large scrambled set implies that orbits of points behave in unpredictable, complex way.
More arguments come from the theory of interval transformations, in view of which it was introduced. For such maps the existence of one Li-Yorke pair implies the existence of an uncountable scrambled set \cite{KuchtaSmital} and it is not very far from implying all other properties that have been called chaotic in this context, see e.g. \cite{Ruette}. In general, Li-Yorke chaos has been proved to be a necessary condition for many other {\it chaotic} properties to hold. A nice survey of properties of Li-Yorke chaotic systems can be found in \cite{BHS}.

%Furstenberg \cite[p. 38]{Fur67} chose to call a dynamical system \emph{deterministic} if its \emph{topological entropy} vanishes. 

A crucial feature of the chaotic behavior of a dynamical system is exponential growth of the number of distinguishable orbits. This happens if and only if the topological entropy of the system is positive. In fact positivity of topological entropy turned out to be an essential criterion of chaos \cite{GW93}.
This choice comes from the fact that the future of a deterministic (zero entropy) dynamical system can be predicted if its past is known (see \cite[Chapter 7]{Weiss}) and  positive entropy is related to randomness and chaos.

For every dynamical system over a compact phase space, we can define a number $h(f)\in[0,\infty]$ called the \emph{topological entropy} of transformation $f$. This quantity was first introduced by Adler, Konheim and McAndrew \cite{AKM} as the topological counterpart of metric (and Shannon) entropy.
%Figuratively speaking topological entropy is an exponential ratio of the growth of number of orbits distinguishable under the finite precision of observation.

For a given positive integer $n$ we define the $n$-th Bowen-Dinaburg metric on $X$, $\rho_n^f$ as \[\rho_n^f(x,y)=\max_{0\leq i<n} dist (f^i(x),f^i(y)).\]

We say that the set $E$ is $(n,\varepsilon)$-separated if $\rho_n^f(x,y)>\varepsilon$ for any distinct $x,y\in E$ and  we denote by $s(n,\varepsilon,f)$ the cardinality of the most numerous $(n,\varepsilon)$-separated set for $(X,f)$.

\begin{definition}
The  topological entropy of $f$ is defined as \[h(f)=\lim_{\varepsilon\searrow 0}\limsup_{n\to\infty}\frac 1n \log s(n,\varepsilon,f).\]
\end{definition} 

We begin with the intuitive explanation of the idea. 
Let us assume that we observe the dynamical system with the precision $\varepsilon>0$, that is, we can distinguish any two points only if they are apart by at least $\varepsilon$. Then, %we are able to see only as many points as is the cardinality of the biggest $(1,\varepsilon)$-separated set. Therefore,
 after $n$ iterations we will see at most $s(n,\varepsilon,f)$ different orbits. If transformation $f$ is mixing points, then $s(n,\varepsilon,f)$ will grow. Taking upper limit over $n$ will give us the asymptotic exponential  growth rate of number of (distinguishable) orbits,  and going with $\varepsilon$ to zero will give us the quantity which can be treated as a measure of exponential speed, with which the number of orbits grow (with $n$). Thus, as Li-Yorke chaos tells us if there is chaos in the system, the topological entropy tells us {\it how much} of chaos we have.

Both positive topological entropy and Li-Yorke chaos are local properties; in fact, entropy depends only on a specific subset of the phase space and is concentrated on the set of so-called nonwandering points \cite{Bow70}.  The question whether positive topological entropy implies
Li-Yorke chaos %or DC$2$ distributional chaos 
remained open for some time, but eventually it was shown to be true; see \cite{BGKM}. %The consequences of positive topological entropy for dynamics of pairs and tuples were also examined in \cite{BH,HLY}.
%Distributional chaos of type $2$ implies Li-Yorke chaos by definition. The converse implication is not true, because if $X=[0,1]$, then by \cite{SS} distributional chaos DC$3$ is equivalent to the positive topological entropy, and 
On the other hand, there are Li-Yorke chaotic interval maps with zero topological entropy (as was shown independently by Sm\'{\i}tal \cite{Smital} and Xiong \cite{Xiong}). %Piku{\l}a proved that there is no connection between positive topological entropy and DC$1$ distributional chaos \cite{Pikula}, therefore Li-Yorke chaos doesn't imply positive topological entropy, either. 
%All above is only a glimpse of the vast literature of the subject. 
For deeper discussion of these matters we refer the reader to the excellent surveys
by Blanchard \cite{Blanchard}, Glasner and Ye \cite{GY}, Li and Ye \cite{LY14} and Ruette's book \cite{Ruette}.

 %--- with explanations

%Quantal response equilibrium (QRE) is arguably one of the most well known solution concept in game theory under the assumption of bounded rationality. It was first introduced by  McKelvey and Palfrey \cite{mckelvey1995quantal,mckelvey1998quantal}.

\subsection{Derivation of the dynamics from multiplicative weights with discounting}
\label{s:MWUd}
In the remainder of this section we show how our model can be derived from Multiplicative Weights Update algorithm that is fitted with discounting of previous costs.
%We consider a two-strategy non-atomic congestion game where the play is driven by Multiplicative Weight Update (MWU) algorithm with discounting of previous costs

%We denote the strategies by $s_1, s_2$.
We recall that the cost of each strategy depends on the fraction of players $x \in [0,1]$ that use strategy $1$
\[
 c_1(x) = \alpha N x, \qquad c_2(1-x) = \beta N (1-x),
\]
where $N > 0$ is the mass of the entire population of players and $\alpha, \beta \geqslant 0$, $\max\{\alpha,\beta\} > 0$, are parameters that differentiate the strategies.

At every step $n \geqslant 0$ the strategies $1$ and $2$ have weights $w_1(n)$ and $w_2(n)$ respectively, where initially $w_1(0)=w_2(0)=1$. 
Then at step $n$ strategy $i$ is chosen by each player with probability $\frac{w_i(n)}{w_i(n)+w_j(n)}$, where $i,j \in \{1,2\}$, $i \neq j$.
%Then at step $n$ the strategy $s_1$ is chosen by each player with probability $\frac{w_1(n)}{w_1(n)+w_2(n)}$ and the strategy $s_2$ is chosen by each player	with probability $\frac{w_2(n)}{w_1(n)+w_2(n)}$. 
Because the population is homogeneous, the fraction of population that uses strategy $1$ at step $n+1$ is
\begin{equation}\label{MWU}
 x_{n+1} = \frac{w_1(n+1)}{w_1(n+1)+w_2(n+1)}.
\end{equation}
The weights are updated as follows
\begin{align*}
 w_1(n+1) &= w_1(0) \cdot (1-\varepsilon)^{\sum\limits_{k=0}^{n} (1-\sigma)^{n-k} c_1(x_k)},\\
 w_2(n+1) &= w_2(0) \cdot (1-\varepsilon)^{\sum\limits_{k=0}^{n} (1-\sigma)^{n-k} c_2(1-x_k)},
\end{align*}
%\[
% w_1(n+1) = w_1(0) \cdot (1-\varepsilon)^{\sum\limits_{k=0}^{n} (1-\sigma)^{n-k} c_1(x_k)}, \qquad
% w_2(n+1) = w_2(0) \cdot (1-\varepsilon)^{\sum\limits_{k=0}^{n} (1-\sigma)^{n-k} c_2(1-x_k)},
%\]
where $\varepsilon \in (0,1)$ is a learning rate and $\sigma \in [0,1]$ is a discount factor that depreciates past costs.
That is, the weight $w_i$ decreases with higher discounted cumulative cost of previous play of strategy $i \in \{1,2\}$.
We express the update rule of the weights in terms of previous-step weights;
\begin{align*}
 w_1(n+1) &= w_1(0) \cdot (1-\varepsilon)^{(1-\sigma)\sum\limits_{k=0}^{n-1} (1-\sigma)^{n-1-k} c_1(x_k)} \cdot (1-\varepsilon)^{c_1(x_n)} \\
 &= \left( w_1(n) \right)^{1-\sigma} \cdot (1-\varepsilon)^{c_1(x_n)}
\end{align*}
and
\begin{align*}
 w_2(n+1) &= w_2(0) \cdot (1-\varepsilon)^{(1-\sigma)\sum\limits_{k=0}^{n-1} (1-\sigma)^{n-1-k} c_2(1-x_k)} \cdot (1-\varepsilon)^{c_2(1-x_n)} \\
 &= \left( w_2(n) \right)^{1-\sigma} \cdot (1-\varepsilon)^{c_2(1-x_n)}.
\end{align*}
%\begin{align*}
% w_1(n+1) &= w_1(0) \cdot (1-\varepsilon)^{(1-\sigma)\sum\limits_{k=0}^{n-1} (1-\sigma)^{n-1-k} c_1(x_k)} \cdot (1-%\varepsilon)^{c_1(x_n)}
% = \left( w_1(n) \right)^{1-\sigma} \cdot (1-\varepsilon)^{c_1(x_n)}, \\
% w_2(n+1) &= w_2(0) \cdot (1-\varepsilon)^{(1-\sigma)\sum\limits_{k=0}^{n-1} (1-\sigma)^{n-1-k} c_2(1-x_k)} \cdot (1-%\varepsilon)^{c_2(1-x_n)}
% = \left( w_2(n) \right)^{1-\sigma} \cdot (1-\varepsilon)^{c_2(1-x_n)}.
%\end{align*}
Then from \eqref{MWU} we have that
\begin{equation}\label{xn_with_weights_and_epsilon}
 x_{n+1} = \frac{\left( w_1(n) \right)^{1-\sigma} \cdot (1-\varepsilon)^{c_1(x_n)}}{\left( w_1(n) \right)^{1-\sigma} \cdot (1-\varepsilon)^{c_1(x_n)} + \left( w_2(n) \right)^{1-\sigma} \cdot (1-\varepsilon)^{c_2(1-x_n)}}.
\end{equation}
%\begin{align*}
% x_{n+1} &= \frac{\left( w_1(n) \right)^{1-\sigma} \cdot (1-\varepsilon)^{c_1(x_n)}}{\left( w_1(n) \right)^{1-\sigma} \cdot (1-\varepsilon)^{c_1(x_n)} + \left( w_2(n) \right)^{1-\sigma} \cdot (1-\varepsilon)^{c_2(1-x_n)}} \\
% &= \frac{\left( \frac{w_1(n)}{w_1(n)+w_2(n)} \right)^{1-\sigma} \cdot \exp\left[ -\log\left( \frac{1}{1-\varepsilon} \right) \cdot c_1(x_n) \right]}{\left( \frac{w_1(n)}{w_1(n)+w_2(n)} \right)^{1-\sigma} \cdot \exp\left[ -\log\left( \frac{1}{1-\varepsilon} \right) \cdot c_1(x_n) \right] + \left( \frac{w_2(n)}{w_1(n)+w_2(n)} \right)^{1-\sigma} \cdot \exp\left[ -\log\left( \frac{1}{1-\varepsilon} \right) \cdot c_2(1-x_n) \right]} \\
% &= \frac{x_n^{1-\sigma} }{x_n^{1-\sigma} + \left( 1-x_n \right)^{1-\sigma} \cdot \exp\left[ N\log\left( \frac{1}{1-\varepsilon} \right) \cdot \left( \alpha x_n - \beta(1-x_n) \right) \right]}.
%\end{align*}
Note that $(1-\varepsilon)^{c_i(\cdot)} = \exp\left[ -\log\left( \frac{1}{1-\varepsilon} \right) \cdot c_i(\cdot) \right]$. Therefore, by dividing the numerator and the denominator of \eqref{xn_with_weights_and_epsilon} by %factor
$\left( w_1(n) + w_2(n) \right)^{1-\sigma} \cdot (1-\varepsilon)^{c_1(x_n)}$ we get
\[
 x_{n+1} = \frac{x_n^{1-\sigma} }{x_n^{1-\sigma} + \left( 1-x_n \right)^{1-\sigma} \cdot \exp\left[ N\log\left( \frac{1}{1-\varepsilon} \right) \cdot \left( \alpha x_n - \beta(1-x_n) \right) \right]}.
\]
Then by denoting $a = (\alpha+\beta)N\log\left( \frac{1}{1-\varepsilon} \right)$ and $b = \frac{\beta}{\alpha+\beta}$ we obtain the update rule \eqref{mwu}.
%\[
%x_{n+1} = \frac{x_n^{1-\sigma} }{x_n^{1-\sigma} + \left( 1-x_n \right)^{1-\sigma} \exp\left( a \left( x_n - b \right) \right)}.
%\]

\section{Results}
\subsection{One dimensional dynamics for $\sigma \in [0,1]$}

We are interested in discrete dynamical system on the unit interval $[0,1]$
\begin{equation}\label{dynamics}
 x_{n+1} = \fabs(x_n) = \frac{x_n^{1-\sigma}}{x_n^{1-\sigma}+(1-x_n)^{1-\sigma}\exp (a(x_n-b))},
\end{equation}
where $a > 0$, $b \in (0,1)$, $\sigma\in [0,1]$.

As we have one to one correspondence between the ratio of agents choosing first resource in the game $x$ and mixed strategy $(x,1-x)$, we will use the simplification saying that the fixed point of $f_{ab\sigma}$ is a Nash/perturbed equilibrium of the game.

 For $\sigma = 1$ the map $f_{ab\sigma}$ is decreasing and has one equilibrium $\ol{x}\in (0,1)$.
To study dynamics of $\fabs$ for $\sigma<1$ one can look at the derivative of $\fabs$ for $x\in (0,1)$, which is given by
\begin{equation} \label{derv1}
 f'_{ab\sigma}(x)= \frac{x^{-\sigma}(1-x)^{-\sigma}\exp(a(x-b))\left[ 1-\sigma - ax(1-x) \right]}{\left[ x^{1-\sigma}+(1-x)^{1-\sigma} \exp(a(x-b)) \right]^2}.
\end{equation}
The derivative of $\fabs$ can be written in an equivalent form
\begin{equation} \label{derv2}
 f'_{ab\sigma}(x)= f_{ab\sigma}(x) \big( 1-f_{ab\sigma}(x) \big) \left( \frac{1-\sigma}{x(1-x)} - a \right).
\end{equation}

 When $\sigma \in [0,1)$ the map $\fabs$ has three equilibria: 0, 1, and $\ol{x}\in (0,1)$. The unique equilibrium in $ (0,1)$ usually depends on $a, b$ and $\sigma$. When $\sigma<1$, the derivative of $\fabs$ is infinite at 0 and 1,
% so 0 and 1 are repelling for all $a>0$, $b\in (0,1)$ and $\sigma\in [0,1)$. 
therefore, the fixed points 0 and 1 are repelling independently of values of intensity of choice $a>0$, discount factor $\sigma \in [0,1)$ and Nash equilibrium of the game $b \in (0,1)$. 

%The map $\fabs$ has 3 fixed points: 0, 1, and $\ol{x}\in (0,1)$. The derivative of $\fabs$ is infinite at 0 and 1. Thus,  0 and 1 are repelling for all $a>0$ and $b\in (0,1)$ and $\sigma\in [0,1]$. The unique fixed point in $\ol{x}\in (0,1)$ usually depends on $a, b$ and $\sigma$. 

From \eqref{derv2} $\fabs$ is a homeomorphism as long as $a\in (0,4(1-\sigma)]$, and when $a>4(1-\sigma)$ the map $\fabs$ is  bimodal with critical points:
\[
 c_l = \frac{1-\sqrt{1-\frac{4(1-\sigma)}{a}}}{2}, \qquad c_r = \frac{1+\sqrt{1-\frac{4(1-\sigma)}{a}}}{2}.
\]

Moreover, from \eqref{derv2} we have that
\begin{equation}\label{dffixedpoint}
 f'_{ab\sigma}(\ol{x}) = a \ol{x}^2 - a \ol{x} + 1-\sigma =1- \sigma - a\ol{x}(1-\ol{x}).
\end{equation}

\subsection{Conjugate map $F$}
\begin{remark}\label{conjugacy}
Set $y=\frac 1a\log\frac{1-x}x$ and $F(y)=\frac 1a\log\frac{1-f_{ab\sigma}(x)}{f_{ab\sigma}(x)}$. Then
%$e^{ay}=\frac{1-x}x$, so $x=\frac1{e^{ay}+1}$. We get
%\[
%\frac{1-f_{ab\sigma}(x)}{f_{ab\sigma}(x)}=\frac{(1-x)^{1-\sigma}\exp(a(x-b))}{x^{1-\sigma}}=
%\left(\frac{1-x}x\right)^{1-\sigma}\exp(a(x-b)),
%\]
%so
\begin{equation} \label{F}
F(y)=(1-\sigma) y+\frac1{e^{ay}+1}-b.
\end{equation}
The map $x\mapsto\frac 1a\log\frac{1-x}x$ is a diffeomorphism from $(0,1)$
onto $\R$. Therefore $f_{ab\sigma}$ on $(0,1)$ is smoothly conjugate to $F$ on
$\R$.\footnote{The map $F$ depends on three parameters, but in order to simplify the
notation we will not mark this dependence.} 
%Studying dynamics of $F$ is usually simpler than $f_{ab\sigma}$. Thus, we will repeatedly look at the dynamics of our system through the lenses of the map $F$.
\end{remark}

As the derivatives of $f_{ab\sigma}$ at 0 and 1 are infinite, we will often make use of the map $F$ (given by \eqref{F}) which by Remark \ref{conjugacy} is topologically conjugate to $f_{ab\sigma}$. This means that instead of investigating the dynamics of $f_{ab\sigma}$, we
may investigate the dynamics of $F$. It is worth adding that studying the dynamics of $F$ is usually simpler than for $f_{ab\sigma}$. Thus, we will repeatedly look at the dynamics of our system through the lenses of the map $F$.

The conjugate map $F$ introduced in \eqref{F} has some nice properties, for instance it has negative Schwarzian derivative for $a>4(1-\sigma)$. This is important as the dynamics is fairly
regular if the map has negative Schwarzian derivative. Recall that the
Schwarzian derivative is given by
\[
SF=\frac{F'''}{F'}-\frac32\left(\frac{F''}{F'}\right)^2.
\]

\begin{lemma}\label{ls}
If $a>4(1-\sigma)$ then Schwarzian derivative of $F$ is negative.
\end{lemma}

\begin{proof}
For simplicity, let us use notation $t=e^{ax}$. Elementary calculations
give us
\begin{align*}
F'(x)   &= 1-\sigma-a\frac{t}{(1+t)^2},\\
F''(x)  &= a^2\frac{t^2-t}{(1+t)^3},\\
F'''(x) &= a^3\frac{-t^3+4t^2-t}{(1+t)^4}.
\end{align*}

Schwarzian derivative of $F$ is negative if and only if
$2F'F'''-3(F'')^2<0$. From our formulas we get
\[
2F'F'''-3(F'')^2=\big(2(1-\sigma)(-t^2+4t-1)-at\big)\frac{a^3t}{(1+t)^4}.
\]
We have
\[
2(1-\sigma)(-t^2+4t-1)-at=2(1-\sigma)\big(2t-(t-1)^2\big)-at\le
4(1-\sigma) t-at=(4(1-\sigma)-a)t,
\]
so if $a>4(1-\sigma)$ then $SF<0$.
\end{proof}

As
\[
F'(x)=1-\sigma-\frac{ae^{ax}}{(e^{ax}+1)^2}.
\]
 either $F$ is strictly increasing, or it is bimodal,
with $F$ increasing on the left and right laps, and decreasing on the
middle lap. By Lemma \ref{ls} if $F$ is bimodal, then it
has negative Schwarzian derivative.

\subsection{Properties of the interior equilibrium}

 %Now we will focus on the properties of $\ol{x}$, a unique fixed point of $f_{ab\sigma}$ in $(0,1)$ for $\sigma \in [0,1]$.% It is a quantal response equilibrium of our game.
As the fixed points 0 and 1 are always repelling, we focus our considerations on $\ol{x}$ -- the interior fixed point of the dynamics~\eqref{dynamics}. By Definition \ref{pertdef} and equation \eqref{perteq} we have the following:

\begin{remark} \label{xpe}
Let $\sigma \in [0,1]$. The interior fixed point of $f_{ab\sigma}$ is a perturbed equilibrium. In particular, when $\sigma = 0$, the interior fixed point of $f_{ab\sigma}$ is a Nash equilibrium.
\end{remark}
%Lemma \ref{xpe} follows from the fact that the dynamics induced by the map $\fabs$ cn be obtained  from  \eqref{perturbedBR} with $R(x)=x\log x+(1-x)\log (1-x)$.
%\begin{align*}
% x_{n+1} = \min_{x \in (0,1)} \Bigg[ &\ln\frac{1}{1-\varepsilon} \sum_{t=0}^n (1-\sigma)^{n-t} \left( c_1(x_t) \cdot x + c_2(1-x_t) \cdot (1-x) \right)\\
%  &+ x\log x + (1-x) \log(1-x) \Bigg].
%\end{align*}

%By Lemma \ref{xpe} the unique interior fixed point is a perturbed equilibrium of the game.

We now describe the location of the interior perturbed equilibrium $\ol{x}$ and its monotonic convergence to either $1/2$ or Nash equilibrium $b$. Firsf we show an auxiliary lemma.

\begin{lemma}\label{monotonic}
If $b \in (0,1/2)$, then $\oxa<0$; if $b \in (1/2,1)$, then $\oxa>0$.
In particular, if $b\in(0,1)\setminus\{1/2\}$ then $(1-2x)\oxa<0$.
\end{lemma}

\begin{proof}
The equation $\fabs(\ol{x}) = \ol{x}$ is equivalent to
\begin{equation}\label{fp11}
 \ol{x} = b + \frac{\sigma}{a} \log\left( \frac{1-\ol{x}}{\ol{x}} \right).
\end{equation}
By~\eqref{fp11}, we have
\[
a\ox=ab+\sigma\log\frac{1-\ox}{\ox}.
\]
Take the total derivtive of both sides with respect to $a$. We get
\[
\ox+a\oxa=b-\frac{\sigma}{\ox(1-\ox)}\cdot\oxa.
\]
Therefore
\begin{equation}\label{e11}
\left(a+\frac{\sigma}{\ox(1-\ox)}\right)\oxa=b-\ox.
\end{equation}
If $b \in (0,1/2)$, then $b-\ox<0$; if $b \in (1/2,1)$ then $b-\ox>0$.
This completes the proof.
\end{proof}

\begin{proposition} \label{qreolx}
Let $\sigma \in (0,1]$, and let $\ol{x} \in (0,1)$ be a (unique) perturbed equilibrium. Then $\ol{x}$ lies between $1/2$ and $b$. Moreover, when the intensity of choice $a$ tends to zero, $\ol{x}$ tends monotonically to $1/2$, while as intensity of choice tends to infinity, $\ol{x}$ converges monotonically to Nash equilibrium $b$. Finally, if $\sigma=0$, then $\ol{x}=b$ is the unique Nash equilibrium for all $a>0$.
%\item $\lim\limits_{a \to 0^+} \ol{x} = \frac 12$, $\lim\limits_{a \to \infty} \ol{x} = b$.
\end{proposition}

\begin{proof}%[Proof of Proposition \ref{qreolx}]

We begin with the first assertion. Let $b < \frac 12$. If $\ol{x} \in \left( 0, \frac 12 \right]$, then $\log\left( \frac{1-\ol{x}}{\ol{x}} \right) \geqslant 0$, and therefore, by \eqref{fp11} we have that $\ol{x} \geqslant b$. Otherwise $\ol{x} \in \left( \frac 12, 1 \right)$, but then $\log\left( \frac{1-\ol{x}}{\ol{x}} \right) < 0$, and by \eqref{fp11} we get $\ol{x} < b < \frac 12$, a contradiction. Thus, $\ol{x} \in [b,\frac 12]$. The proof of the case $b > \frac 12$ is analogous.

We proceed with the second assertion. We put \eqref{fp11} in the equivalent form
\begin{equation}\label{fp2}
 a(\ol{x}-b) = \sigma \log\left( \frac{1-\ol{x}}{\ol{x}} \right).
\end{equation}
Denote
\[
 L(a) = a(\ol{x}-b) \qquad \text{and}\qquad R(a) = \sigma \log \left( \frac{1-\ol{x}}{\ol{x}} \right).
\]
Note that $\lim\limits_{a \to 0^+} L(a) = 0$. By \eqref{fp2} $L(a)=R(a)$ for all $a$. Thus, $\lim\limits_{a \to 0^+} R(a) = 0$. For $\sigma \in (0,1]$ this last equation holds if and only if $ \lim\limits_{a \to 0^+}\ol{x} = \frac 12$.

Because $\ol{x} \in [b,\frac 12]$ for $b<\frac 12$ and $\ol{x} \in [\frac 12,b]$ for $b>\frac 12$, we have that $\lim\limits_{a \to \infty} R(a) < \infty$. By \eqref{fp2} we have $\lim\limits_{a \to \infty} L(a) < \infty$. This last inequality holds if and only if $\lim\limits_{a \to \infty} \ol{x} = b$.

The fact that the convergence of $\ol{x}$ is monotonic follows from Lemma \ref{monotonic}.
\end{proof}

Proposition \ref{qreolx} guarantees that the perturbed equilibrium is bounded by $\frac 12$ and $b$. Moreover, it describes two extreme cases. When the intensity of choice tends to zero, the perturbed equilibrium $\ol{x}$ approaches the case when an agent  is indifferent about his payoff and thus which resource to choose. As both choices are equally likely, the split $(\frac12,\frac12)$ is chosen. On the other hand, if intensity of choice tends to infinity, then a small historical advantage of a given choice causes that choice to be more probable. Then the perturbed equilibrium approaches Nash equilibrium. %This last result agrees with the theory of perturbed equilibria.

\begin{remark}
Proposition \ref{qreolx} and equation \eqref{fp11} imply that as $\sigma$ converges to $0$, the perturbed equilibrium $\ol{x}$ converges to Nash equilibrium $b$. 
\end{remark}

%Since we know that $\overline{x}$ is a QRE it is not suprising 
We next study the convergence of trajectories of the dynamics \eqref{dynamics} to the perturbed equilibrium.

\begin{theorem} \label{afp}
Let $\sigma \in [0,1)$. As long as the perturbed equilibrium $\ol{x}$  is attracting, it attracts all trajectories of points from $(0,1)$. For $\sigma=1$ the perturbed equilibrium attracts also trajectories of $0$ and $1$.
\end{theorem}

%In the proof of Theorem \ref{afp} we work with the conjugate map $F$ from \eqref{F}. For $F$ we show that if a fixed point of $F$ is attracting, then it is globally attracting.

%\subsection{Proof of Theorem \ref{afp}}

Theorem \ref{afp} guarantees that local stability of the perturbed equilibrium imply global convergence to this perturbed equilibrium. In other words, starting from any initial condition, that is any mixed strategy profile $(x,1-x)$, the system will converge to $(\overline{x},1-\overline{x})$.\footnote{It is worth mentioning that existence of attracting perturbed equilibrium or even attracting Nash equilibrium does not exclude possibility of chaotic behavior. For instance, Follow the Regularized Leader algorithm, admits coexistence of attracting Nash equilibrium and chaos \cite{BCFKMP21}.} Therefore, the description of the dynamics of the game is simple as long as $\ol{x}$ is attracting.

\begin{proof}[Proof of Theorem \ref{afp}]
We are going to show that if $\ol{x}$ is attracting, then it attracts all points from $(0,1)$. To this aim we will work on the conjugate map $F$.

\begin{lemma}\label{l2a} The map $F$ has a unique fixed point $\fp$. 
If the trajectories of all points $x<\fp$ are attracted to $\fp$, then
the trajectories of all points of $\R$ are attracted to $\fp$.
Similarly, if the trajectories of all points $x>\fp$ are attracted to
$\fp$, then the trajectories of all points of $\R$ are attracted to
$\fp$.
\end{lemma}

\begin{proof}
If $x$ is sufficiently large, then $F(-x)>-x$ and $F(x)<x$. Therefore,
$F$ has a fixed point. Since $F'<1-\sigma<1$, by the Mean Value Theorem,
$F$ cannot have two distinct fixed points. We will denote the fixed point of $F$ by $\fp$. Obviously $\fp=\log\frac{1-\ol{x}}{\ol{x}}$.

Assume that there is a point of $\R$, whose trajectory is not
attracted to $\fp$. Since both $-\infty$ and $\infty$ are repelling,
by \cite{SKSF}, $F$ has a periodic orbit of period 2. If the trajectories
of all points $x<\fp$ (respectively, $x>\fp$) are attracted to $\fp$,
this periodic orbit has to lie entirely to the right (respectively,
left) of $\fp$. Thus, there is a fixed point to the right
(respectively, left) of $\fp$, a contradiction.
\end{proof}

\begin{lemma}\label{t1a}
If the fixed point of $F$ is attracting, then it is globally
attracting.
\end{lemma}

\begin{proof}
If $F$ is strictly increasing, then it does not have a periodic orbit
of period 2, so $\fp$ is globally attracting.

Assume that $F$ is bimodal. If $\fp$ belongs to the left or right lap,
then by Lemma~\ref{l2a}, $\fp$ is globally attracting. Assume that
$\fp$ belongs to the interior of the middle lap. Since by Lemma \ref{ls} the Schwarzian
derivative of $F$ is negative, then the interval joining $\fp$ with
one of the critical points of $F$ is in the basin of attraction $A$ of
$\fp$. We may assume that this critical point is the left one, $c_-$.
There is a unique point $y<c_-$ such that $F(y)=\fp$. Then
$F([y,c_-])=F([c_-,\fp])\subset A$, so $[y,\fp]\subset A$. For every
point $x<y$ we have $x<F(x)<\fp$. Therefore, the trajectory of $x$
increases as long as it stays to the left of $y$. Since there are no
fixed points to the left of $y$, the trajectory has to enter $[y,\fp]$
sooner or later. This proves that $(-\infty,\fp]\subset A$, so by
Lemma~\ref{l2a}, $\fp$ is globally attracting.
\end{proof}

Now as $F$ is a conjugate map for $f$ we obtain Theorem \ref{afp}.
\end{proof}

From Proposition \ref{qreolx} we know that when $\sigma>0$ the perturbed equilibrium $\ol{x}$ will approach $1/2$ for small values of $a$ and will approach Nash equilibrium for sufficiently large intensity of choice. Thus, one may be interested in choosing large values of the parameter $a$. But does such behavior will result in (approximate) convergence of trajectories of the system to Nash equilibrium? Theorem \ref{afp} guarantees convergence as long as $\ol{x}$ is attracting. Nevertheless, we show that increasing the intensity of choice will result in losing stability of $\ol{x}$, and thus, the system will become unstable.

%Moreover, although it is impossible to provide an explicit condition for stability of $\ol{x}$ because it depends on $a$, $b$ and $\sigma$, one can show that when we increase intensity of choice, then once the stability is lost it is lost forever. In addition, for a fixed $a>0$ as $\sigma$ increases, the region of stability of $\ol{x}$ shrinks.
%Therefore, the increase in memory loss will destabilize the system for smaller values of the intensity of choice.

\begin{proposition} \label{QREstability}
There exists $a_0>0$ such that $\ol{x}$ is attracting for $a<a_0$, and $\ol{x}$ is repelling for $a>a_0$. Moreover, the threshold $a_0$ is decreasing with respect to the discount factor $\sigma$. 
%The instability of the quantal response equilibrium $\ol{x}$ appears earlier (with smaller intensity of choice $a$) as agents forget more from the previous plays ($\sigma$ decreases).
\end{proposition}

%The uniqueness of the threshold $a_0$ follows from the fact that $f'_{ab\sigma}(\ol{x})$ is decreasing as a function of intensity of choice.  The proof of the second assertion relies on the observation that for $0\leq \sigma_1<\sigma_2\leq 1$ values of $f_{ab\sigma_2}$ surpass values of $f_{ab\sigma_1}$ when $x<1/2$, while for $x>1/2$ we have the opposite inequalities.

%The instability of the quantal response equilibrium $\ol{x}$ appears earlieras agents forget more from the previous plays ($\sigma$ decreases).

\begin{proof}%[Proof of Proposition \ref{QREstability}]
For the proof of the first assertion we show that if $b\in(0,1)$, then $\fabs'(\ol{x})$ is decreasing as a
function of $a$.

Assume first that $b\ne 1/2$. Multiply both sides of~\eqref{e11} by
$1-2\ox$:
\[
(1-2\ox)\left(a+\frac{\sigma}{\ox(1-\ox)}\right)\oxa=(1-2\ox)(b-\ox)
=b(1-b)+(b-\ox)^2-\ox(1-\ox).
\]
From this and~\eqref{dffixedpoint}, we get
\[
\frac{d\fabs'(\ox)}{da}=-\ox(1-\ox)-a(1-2\ox)\oxa=
\frac{\sigma}{\ox(1-\ox)}(1-2\ox)\oxa-b(1-b)-(b-\ox)^2.
\]
In view of Lemma~\ref{monotonic}, this is negative.

If $b=1/2$, then $\ox=1/2$, and $\fabs'(\ox)=1-\sigma-a/4$, so also
$\fabs'(\ol{x})$ is decreasing as a function of $a$.

We now move to the second assertion which states that the threshold $a_0$ is decreasing with respect to $\sigma$.

Let  $0\leq \sigma_1<\sigma_2\leq 1$. Then
\begin{equation} \label{a}
 f_{ab\sigma_1}(x) < f_{ab\sigma_2}(x) \Longleftrightarrow x < \frac 12,\qquad
 f_{ab\sigma_1}(x) > f_{ab\sigma_2}(x) \Longleftrightarrow x > \frac 12
\end{equation}
These inequalities follow from the fact that
\[ f_{ab\sigma_1}(x) < f_{ab\sigma_2}(x) \;\;\Leftrightarrow\;\; 1+\left(\frac{1-x}{x}\right)^{1-\sigma_2}\exp (a(x-b))<1+\left(\frac{1-x}{x}\right)^{1-\sigma_1}\exp (a(x-b))\]
Thus,
\[f_{ab\sigma_1}(x) < f_{ab\sigma_2}(x) \;\;\Longleftrightarrow\;\; \left(\frac{1-x}{x}\right)^{\sigma_2-\sigma_1}>1\]
and
\[f_{ab\sigma_1}(x) > f_{ab\sigma_2}(x) \;\;\Longleftrightarrow\;\; \left(\frac{1-x}{x}\right)^{\sigma_2-\sigma_1}<1\]
which implies \eqref{a}.

Let $b< 1/2$. Then $\ol{x}\in [b,1/2]$. From \eqref{a} we have \[f_{ab\sigma_2}(\ol{x}_{\sigma_1})>f_{ab\sigma_1}(\ol{x}_{\sigma_1})=\ol{x}_{\sigma_1}.\] Thus, once more from \eqref{a}, we infer that $1/2\geq \ol{x}_{\sigma_2}>\ol{x}_{\sigma_1}\geq b$.

Therefore, from \eqref{dffixedpoint} and the fact that
a term $z(1-z)$ increases if and only if the distance between $z$ and $\frac 12$ decreases,  we have $f_{ab\sigma_1}'(\ol{x}_{\sigma_1})>f_{ab\sigma_2}'(\ol{x}_{\sigma_2})$ for any given $a>0$.
Similar reasoning can be performed for $b>1/2$. Since for $b=1/2$ the only difference in the reasoning is that $\ol{x}_{\sigma_1}=\ol{x}_{\sigma_2}=1/2$, we obtain that for every $b\in (0,1)$ and $a>0$ \[f_{ab\sigma_1}'(\ol{x}_{\sigma_1})>f_{ab\sigma_2}'(\ol{x}_{\sigma_2}).\]
Thus (as the derivative at the fixed point cannot be greater than one), the instability at $\ol{x}_{\sigma_2}$ arises for smaller values of $a$ than for $\ol{x}_{\sigma_1}$.
\end{proof}

Proposition \ref{QREstability} implies that perturbed equilibrium is stable for sufficiently small intensity of choice. Then, once the stability is lost, with increasing intensity of choice, it will remain unstable.  Moreover, for a fixed $b$ as $\sigma$ increases, the region of stability of $\ol{x}$ shrinks.
Therefore, the increase of discount factor (memory loss) will destabilize the system for smaller values of intensity of choice.

%as agents discount/forget more the instability of the system will appear earlier  (with smaller intensity of choice $a$). Thus, convergence reasoning will be useful for smaller range of values of intensity of choice.

To sum up the findings of Theorem \ref{afp} and Proposition \ref{QREstability}, we have that for small intensity of choice the perturbed equilibrium attracts all trajectories of the system, so starting from any initial state (other than the case where the entire population chooses a pure strategy) the system will converge to $(\ol{x}, 1-\ol{x})$. Then there is a threshold where the perturbed equilibrium loses stability. Therefore, increasing intensity of choice will eventually destabilize the system. This threshold depends on discount (memory loss) factor $\sigma$ in monotonic way --- as more memory is lost ($\sigma$ increases), the instability appears earlier, for smaller intensity of choice.

\begin{proposition} \label{abthreshold}
If $\sigma>0$, then there exists threshold $a^*>0$ such that for $a>a^*$ the perturbed equilibrium $\ol{x}$ is repelling for every $b\in (0,1)$.
If $\sigma=0$, then for any intensity of choice $a$ there exists $b$ (sufficiently close to $0$ or $1$) such that the Nash equilibrium $\ol{x}$ is attracting.
\end{proposition}

\begin{proof}%[Proof of Proposition \ref{abthreshold}]
To show Proposition \ref{abthreshold} it is enough to prove that for $\sigma>0$ the boundary in $(a, b)$ plane between the region of attracting perturbed equilibrium and the region of repelling perturbed equilibrium crosses the levels $b=0$ and $b=1$. %On the other hand, 
For $\sigma=0$ this boundary does not cross the levels $b=0$ and $b=1$, instead it approaches them as intensity of choice increases (see Figures \ref{fig:025}, \ref{fig:075}).

First, we compute $\olx$ from the equation $f'_{ab\sigma}(\olx) = -1$. This equation has two real solutions when $a \geqslant 4(2-\sigma)$
\[
 \olx_1 = \frac 12 \left( 1 - \sqrt{1-\frac{4(2-\sigma)}{a}} \right) \in \big( 0,1/2 \big], \qquad \olx_2 = \frac 12 \left( 1 + \sqrt{1-\frac{4(2-\sigma)}{a}} \right) \in \big[ 1/2,1 \big).
\]
Observe, that these solutions are symmetric
\begin{equation}\label{symmetric_olx}
 1 - \olx_1 = 1 - \frac 12 \left( 1 - \sqrt{1-\frac{4(2-\sigma)}{a}} \right) = \frac 12 \left( 1 + \sqrt{1-\frac{4(2-\sigma)}{a}} \right) = \olx_2.
\end{equation}
Moreover, with $\sigma \in [0,1]$ fixed:
\begin{itemize}
\item $\olx_1$ as a function of $a$ is bijection $\big[ 4(2-\sigma),\infty \big) \longrightarrow \big( 0,\frac 12 \big]$,
\item $\olx_2$ as a function of $a$ is bijection $\big[ 4(2-\sigma),\infty \big) \longrightarrow \big[ \frac 12, 1 \big)$.
\end{itemize}

Second, we insert the solutions of $f'_{ab\sigma}(\olx) = -1$ into the equation for the fixed point of $f_{ab\sigma}$. As a result, we obtain formulas for $b$ as functions of $a$ and $ \sigma$:
\begin{align*}
 b_1(a,\sigma) &= \olx_1 - \frac{\sigma}{a} \log\left( \frac{1-\olx_1}{\olx_1} \right) \in \big[ 0,1/2 \big], \\
 b_2(a,\sigma) &= \olx_2 - \frac{\sigma}{a} \log\left( \frac{1-\olx_2}{\olx_2} \right) \in \big[ 1/2,1 \big].
\end{align*}
The first formula describes the bottom branch of the boundary between the region of stability of the fixed point and the region of attracting periodic orbit of period 2, and the second formula - the upper branch.

By \eqref{symmetric_olx} we obtain the the functions $b_1$ and $b_2$ are also symmetric
\begin{align*}
 1 - b_1(a,\sigma) &= 1 - \olx_1 - \frac{\sigma}{a} \log\left( \frac{1-\olx_1}{\olx_1} \right)^{-1}
 = (1 - \olx_1) - \frac{\sigma}{a} \log\left( \frac{\olx_1}{1-\olx_1} \right) \\
 &= \olx_2 - \frac{\sigma}{a} \log\left( \frac{1-\olx_2}{\olx_2} \right)
 = b_2(a,\sigma).
\end{align*}

We next determine a solution of $b_1(a,\sigma) = 0$ (and by symmetry of $b_2(a,\sigma) = 1$)
\begin{align*}
 &b_1(a,\sigma) = 0 \quad \Longleftrightarrow\quad \frac{\sigma}{a} \log\left( \frac{1-\olx_1}{\olx_1} \right) = \olx_1
 \quad\Longleftrightarrow\quad \frac{\sigma}{a} \olx_2 \log\left( \frac{1-\olx_1}{\olx_1} \right) = \olx_1\olx_2 \\
 &\Longleftrightarrow\quad \frac{\sigma}{a} (1-\olx_1) \log\left( \frac{1-\olx_1}{\olx_1} \right) = \frac{2-\sigma}{a}
 \quad\Longleftrightarrow\quad (1-\olx_1) \log\left( \frac{1-\olx_1}{\olx_1} \right) = \frac{2-\sigma}{\sigma}.
\end{align*}
Define
\[
 g(x) = (1-x) \log\left( \frac{1-x}{x} \right).
\]
Then $g(1/2) = 0$, $\lim\limits_{x \to 0^+} g(x) = \infty$ and $g'(x) = -\log\left( \frac{1-x}{x} \right) - \frac{1}{x} < 0$ for $x \in (0, 1/2]$. Therefore, $g\colon (0, 1/2] \mapsto [0,\infty)$ is bijection. By the fact that $\olx_1$ as a function of $a$ is also bijection we obtain that for and fixed $\sigma \in (0,1]$ there exists a unique $a_1 \in \big[ 4(2-\sigma),\infty \big)$ such that $b_1(a_1,\sigma) = 0$ and $b_2(a_1,\sigma) = 1$.

On the other hand, when $\sigma=0$ the equation $b_1(a,\sigma) = 0$ does not have any solution. Instead
\[
 \lim_{a \to \infty} b_1(a,\sigma) = 0, \qquad \lim_{a \to \infty} b_2(a,\sigma) = 1.
\]

\vspace{0.3cm}

Implications of these results:
\begin{enumerate}
\item If $\sigma \in (0,1]$, then there exists $a^* \geqslant 4(2-\sigma)$ such that for every $a>a^*$ the fixed point $\olx$ is not attracting.
\item If $\sigma=0$, then for any $a>0$ there exists $b \in (0,1)$ (sufficiently close to 0 or sufficiently close to 1) such that the fixed point $\olx=b$  is attracting.
\end{enumerate}
\end{proof}

Proposition \ref{abthreshold} gives another important distinction between no discount (perfect memory) model and discount (memory loss) case. When the intensity of choice is large, then in perfect memory case one can change conditions of the game (differentiate costs of the (pure) strategies) to impose the convergence to perturbed equilibrium. However, once the memory loss affects choices of agents ($\sigma>0$), then for a sufficiently large intensity of choice the system will inevitably become unstable and no change of conditions of the game will stabilize it.

In the remaining part of the article we study exactly how unpredictable the system will become when stability is lost.

\subsection{Logit best response dynamics --- no memory case}

We first discuss the case when there is no memory of previous learning steps, that is when $\sigma=1$. We show that the dynamics in this case is simple (see Proposition \ref{dynsigma0}).

In no memory case we get well-known logit best response dynamics \cite{alos2010logit,blume1993statistical}. %{\color{red} EXPAND (what is the exact connection)}  
In this dynamics, which can be derived from a random utility model, players adopt an action according to a full-support distribution of the logit form, which allocates larger probability to those actions which would deliver (myopically) larger payoffs. It therefore combines the advantage of having a specific theory about the origin of mistakes with the fact that it takes the magnitude of (suboptimal) payoffs fully into account. Noise is incorporated in the specification from the onset, but choices concentrate on best responses as noise vanishes.
For $\sigma=1$ the dynamics is described by the map%our map $\fabs$ is
\[
 f_{ab1}(x) =\frac{1}{1+\exp (a(x-b))}.
\]
This map is decreasing and because $f_{ab1}(0)>0$, $f_{ab1}(1)<1$ it has a unique equilibrium $\ol{x}\in (0,1)$. Monotonicity of $f_{ab1}$ yields that in the discussed case we don't have complicated dynamics. 

\begin{proposition} \label{dynsigma0}  Fix $b\in (0,1)$.
There exists intensity of choice $a_0>0$ such that $\overline{x}$ is attracting as long as $a<a_0$ and repelling when $a>a_0$.
Trajectories of all points from $[0,1]$ converge to the perturbed equilibrium $\overline x$ when $\ol{x}$ is attracting.  Otherwise, that is for $a>a_0$, it has an attracting periodic orbit of period $2$ which attracts trajectories of all points from $[0,1]$ except countably many points whose trajectories fall into $\overline{x}$.
\end{proposition}

%The proof of this result follows from the fact that the map $f^2_{ab1}$ is increasing. The monotonicity of $f^2_{ab1}$ excludes existence of periodic orbits (other than fixed points) of $f^2_{ab1}$, and as a result, we obtain that $f_{ab1}$ does not have any periodic orbit of period greater than $2$. The threshold $a_0$ is one from Proposition \ref{QREstability}.

\begin{proof}%[Proof of Proposition \ref{dynsigma0}]
First, from \eqref{F} the map
\[F(x)=\frac{1}{e^{ax}+1}-b\]
is decreasing.
%First, $\overline{x}$ is attracting as long as $|f_{ab0}'(\overline{x})|<1$. As \[f_{ab0}'(\overline{x})=-\frac{a\exp(a(x-b))}{(1+\exp(a(x-b)))^2}\]
%Because $f_{ab1}$ is decreasing, 
Thus, $F^2$ is increasing. %Therefore, $f^2_{ab1}(x) \geqslant x$ for all $x \in [0,1]$.  This last condition 
This excludes existence of periodic orbits (other than fixed points) of $F^2$. As a result, $F$ does not have any periodic orbit of period greater than $2$. Thus, all trajectories converge to the fixed point $\ol{y}$ or a periodic orbit of period $2$ of $F$.

Values of $F$ are bounded by$-b$ and $1-b$, so $F$ has an attracting invariant interval $[-b,1-b]$.   
Let $a_0$ be a threshold from Proposition \ref{QREstability}. Let $a<a_0$.
Because from Lemma \ref{ls} $F$ has negative Schwarzian derivative and $F$ has no critical points then, by Singer theorem, $(-\infty,\ol{y}]$ or $[\ol{y},\infty)$ has to be attracted by $\ol{y}$. By Lemma \ref{l2a} $\ol{y}$ has to be globally attracting.

 Let $a>a_0$. Notice that 
 \[F^2(x)=x \Longleftrightarrow \frac{1}{1+\exp(aF(x))}-b=x \Longleftrightarrow F(x)=\frac 1a \log \left(\frac{1}{x+b}-1\right).\]
 Therefore, if $\{\gamma_1,\gamma_2\}$ is a periodic orbit of period 2, then 
 \begin{equation} \label{sigma1Fp2}
 \gamma_2=\frac 1a \log \left(\frac{1}{\gamma_1+b}-1\right).
 \end{equation}
Assume that $F$ has two attracting orbits of period 2: $\{\gamma_1',\gamma_2'\}$ and $\{\gamma_1'',\gamma_2''\}$. Without loss of generality we can assume that $\gamma_1'<\gamma_1''$. Then, by \eqref{sigma1Fp2}, we get that $\gamma_2'>\gamma_2''$. Since by Lemma \ref{ls} the Schwarzian
derivative of $F$ is negative, then  in the immediate basin of attraction of each periodic orbit has to be $-b$ or $1-b$.
We may assume that  $[-b,\gamma_1')$   is in the basin of attraction of
$\{\gamma_1',\gamma_2'\}$ and $(\gamma_2'',1-b]$ is in the basin of attraction of $\{\gamma_1'',\gamma_2''\}$. But then $\gamma_2'$ is attracted to $\{\gamma_1'',\gamma_2''\}$. This contradicts existence of two attracting periodic orbits. 
% We may assume that $(-\infty,\gamma_1')$ is attracted.
\end{proof}

Proposition \ref{dynsigma0} narrows down possible long-term behavior of the system for an arbitrary asymmetry of costs $b$ --  starting from any initial mixed strategy the trajectory will converge to the perturbed equilibrium $\overline{x}$ or to the attracting periodic orbit of period 2. The former and the latter behavior depends on the intensity of choice. Thus, losing stability of perturbed equilibrium leads to periodic behavior of the system.
%To sum up, when agents are affected by the maximal possible memory loss the system will either converge to the perturbed equilibrium or will oscillate following attracting periodic orbit of period 2. %  In Section \ref{?} we will show that whether we have the former or the latter case depends on the value of intensity of choice $a$. 
 
\subsection{Memory-dependent behavior}

%In this part we investigate the dynamics induced by $\fabs$.

% If $\sigma>0$, then for $a<4\sigma$ the transformation $\fabs$ is increasing. Thus, %by  \eqref{a} the fixed point $\ol{x}$ attracts every point from $(0,1)$.

%\begin{comment}

\begin{figure}
\label{fig:periodic}
     \centering
          \begin{subfigure}{0.9\textwidth}
         \centering
         \includegraphics[width=\textwidth]{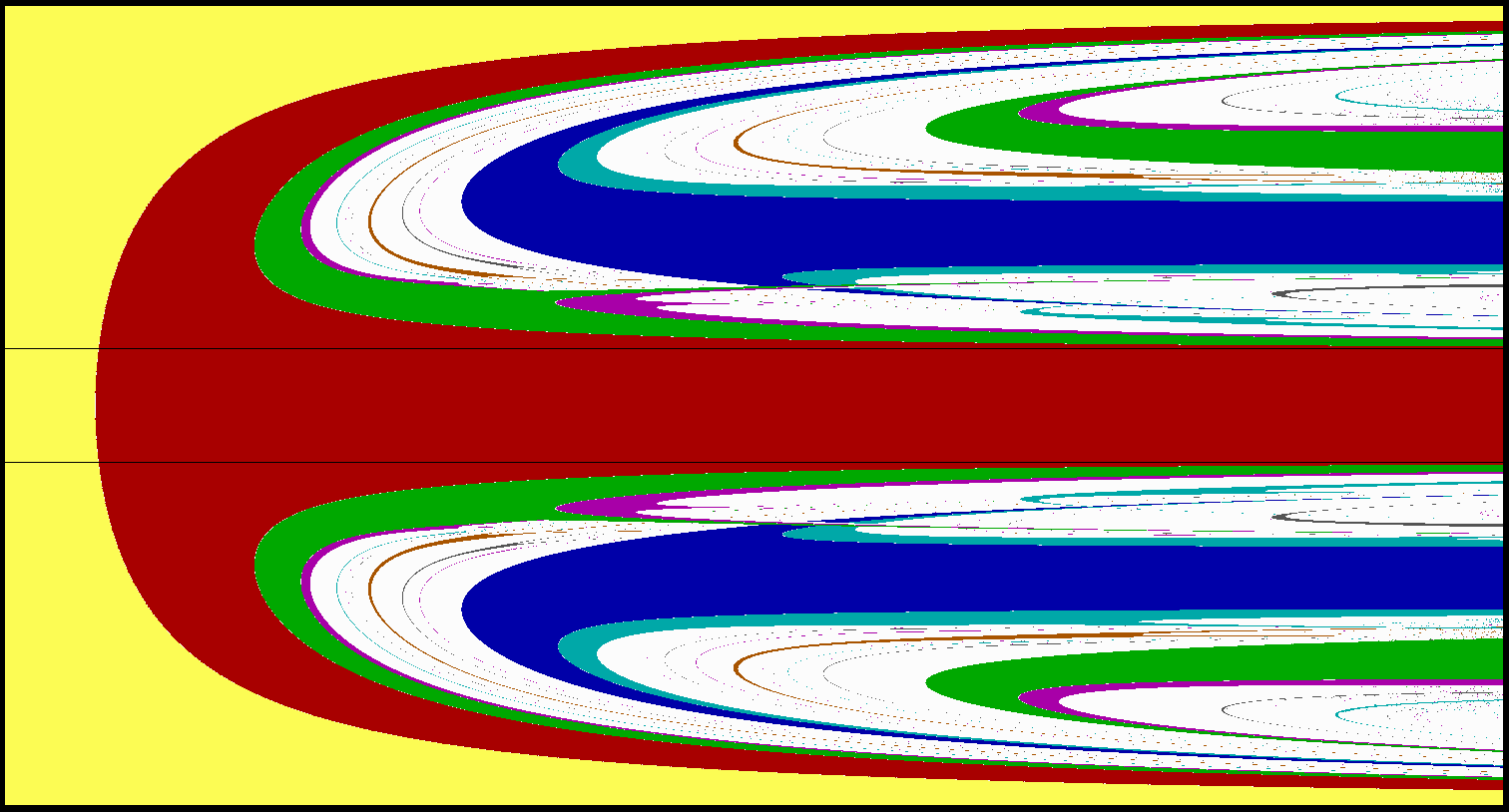}
         \caption{$\sigma=0.25$}
         \label{fig:025}
     \end{subfigure}
%     \hfill
%       \begin{subfigure}[ht]{0.9\textwidth}
%         \centering
%         \includegraphics[width=\textwidth]{periods_F_s_50-16_lines}
%         \caption{$\sigma=0.5$}
%         \label{fig:050}
%     \end{subfigure}
     \hfill
     \begin{subfigure}{0.9\textwidth}
         \centering
         \includegraphics[width=\textwidth]{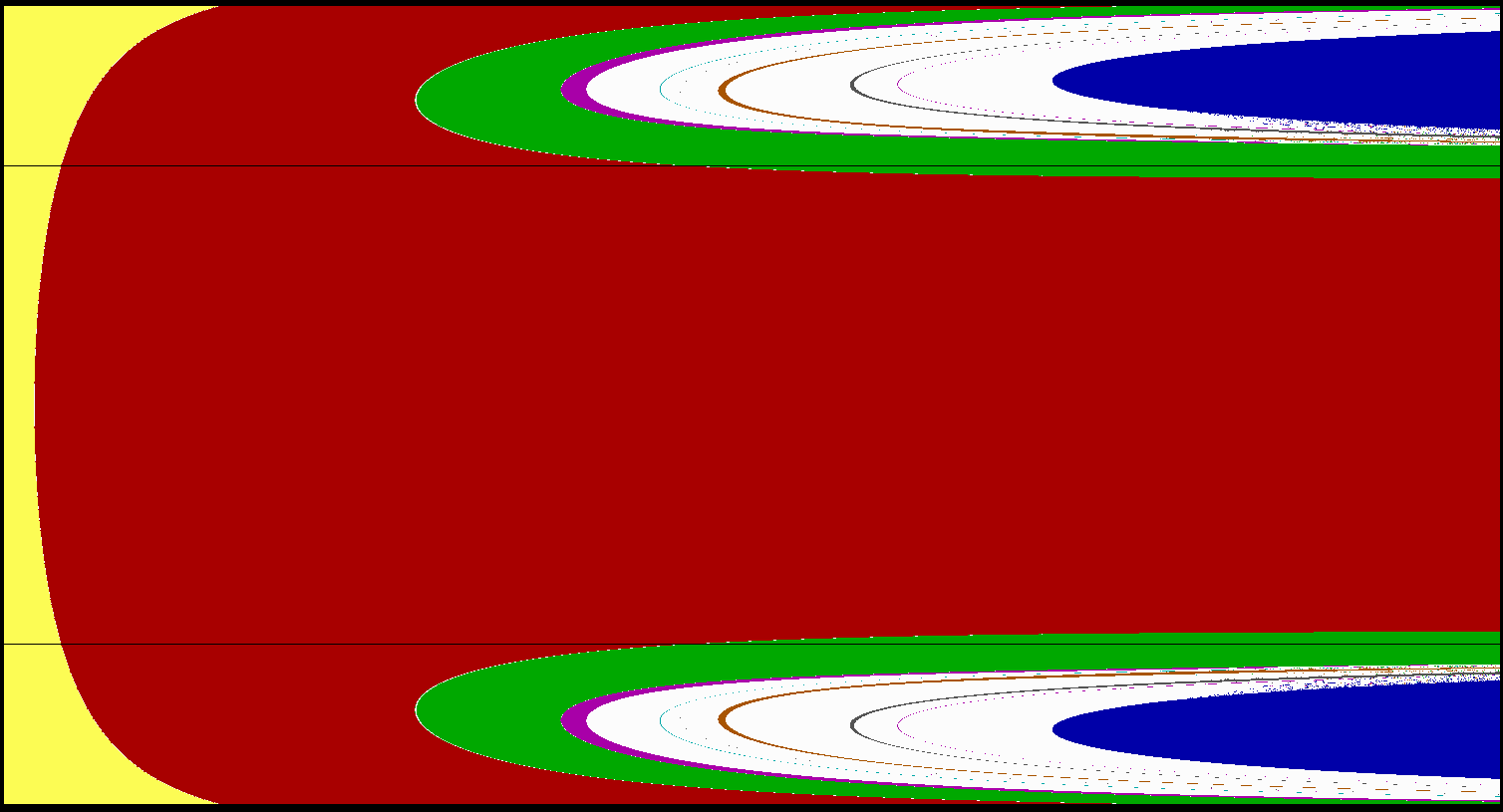}
         \caption{$\sigma=0.75$}
     \end{subfigure}
     \begin{minipage}{0.9\textwidth}
        \caption{\small Period diagrams of the small-period attracting periodic orbits associated with the map $f_{ab\sigma}$ for different values of $\sigma$. The horizontal axes are the intensity of choice $a \in [4,54]$ and the vertical axes are the asymmetry of cost $b \in [0,1]$. The colors encode the periods of attracting periodic orbits as follows: period 1 (fixed point) = {\color{yellow}yellow}, period 2 = {\color{red}red}, period 3 = {\color{blue}blue}, period 4 = {\color{green}green}, period 5 = {\color{brown}brown}, period 6 = {\color{cyan}cyan}, period 7 = {\color{darkgray}darkgray}, period 8 = {\color{magenta}magenta}, and period larger than 8 = white. The equilibrium analysis is only viable when the fixed point $b$ is stable, i.e. when $a \le 2/b(1-b)$. In other region of the phase-space, non-equilibrating dynamics arise and system proceeds through the period-doubling bifurcation route to chaos in the white region. The picture is generated from the following algorithm: 20000 preliminary iterations are discarded. Then a point is considered periodic of period $n$ if $|f_{ab\sigma}^n(x)-x|<10^{-16}$ and it is not periodic of any period smaller than $n$. %Slight asymmetry is caused by the fact that the starting point is the left critical point $x_l = 1/2 - \sqrt{1/4-1/a}$. In addition, for a fixed $a$, as we vary $b$ and penetrate into the chaotic regimes (white) from the outer layers, we numerically observe Feigenbaum's universal route to chaos as discussed below.
       Black lines describe bounds $b=\frac{1-\sigma}{2-\sigma}$, $b=\frac{1}{2-\sigma}$. By Theorem \ref{thm2period} we know what happens for large $a$. We can see from numerical computations that situation might be a little bit more complicated for some values of $a$. There is a possibility of the attracting periodic orbit of period 4.}
       \label{fig:075}
       \end{minipage}
\end{figure}
%\end{comment}

In this subsection we discuss long-term behavior of agents for large intensity of choice when $\sigma \in [0,1)$.

We study how  behavior of the system for large values of intensity of choice $a$ depends on the interplay between memory loss $\sigma$ (property of learning) and difference in costs of resources (property of the game reflected by the value of $b$).
Long-term behavior of agents differs when $b \in \left(\frac{1-\sigma}{2-\sigma},\frac{1}{2-\sigma}\right)$ and when $b \in \left(0, \frac{1-\sigma}{2-\sigma} \right) \cup \left( \frac{1}{2-\sigma}, 1\right)$, that is, it depends on proximity of the costs of resources.

\subsubsection{Long-term behavior for  $b \in \left(\frac{1-\sigma}{2-\sigma},\frac{1}{2-\sigma}\right)$} We begin with the case $b \in \left(\frac{1-\sigma}{2-\sigma},\frac{1}{2-\sigma}\right)$. We investigate the existence of an attracting periodic orbit of period 2 when intensity of choice is large.
%costs of different resources are close, that is when $b$ is close to $1/2$.  

%{\color{red} sketch --- what is for what...}

\begin{theorem}\label{thm2period}
\label{2period}
Let $\sigma\in [0,1)$ be fixed. For a given $b_0\in\left(\frac{1-\sigma}{2-\sigma},\frac12\right)$, there exists $a_1>0$
such that if $a\ge a_1$ and $b\in[b_0,1-b_0]$ then $f_{ab\sigma}$ has an
attracting periodic orbit of period $2$ which attracts trajectories of all points from $(0,1)$, except countably many whose trajectories fall into $\ol{x}$.
\end{theorem}

\begin{corollary} \label{cor2period}
For a fixed $\sigma\in[0,1)$ and $b\in \left(\frac{1-\sigma}{2-\sigma},\frac{1}{2-\sigma}\right)$ there exists $a_1>0$ such that for $a\geq a_1$ trajectories of all points from $(0,1)$ are attracted to the periodic orbit of period $2$, except countably many whose trajectories fall into the interior perturbed equilibrium $\ol{x}$. 
\end{corollary}

Theorem \ref{thm2period} guarantees that when the intensity of choice is large enough, then the system will inevitably converge to an attracting periodic orbit of period 2.
Thus, although the system will not converge, the behavior will be predictable. 
Moreover, the threshold $a_1$ can be chosen in such a way that for a wide variety of levels of asymmetry of cost functions ($b\in [b_0,1-b_0]$) if the intensity of choice crosses this level each system will be attracted (excluding countably many trajectories falling into perturbed equilibrium, but which are almost impossible to get into) to the attracting periodic orbit of period 2. Obviously this attracting periodic orbit has to depend on the values of $a$, $b$ and $\sigma$.

The proof of Theorem \ref{thm2period} relies on careful choice of two disjoint intervals $I_-$ and $I_+$ such that $F(I_-)\subset I_+$ and $F(I_+)\subset I_-$. This last property implies existence of an attracting periodic orbit of period 2. As we deal with a discrete dynamical system we have to take into account that some trajectories may fall into the repelling perturbed equilibrium.
The property of attraction of almost all trajectories follows from existence of an attracting invariant set. We show that all trajectories eventually enter the invariant set and then they either hit the fixed point and stay there, or they are attracted by the periodic orbit.

\begin{proof}[Proof of Theorem \ref{2period}]
We begin with some auxiliary lemmas.

In addition to $F$, let us consider two linear maps,
\[
F_-(x)=(1-\sigma) x+1-b\ \ \ \textrm{and}\ \ \ F_+(x)=(1-\sigma) x-b.
\]
Since $0<\frac1{e^{ax}+1}<1$, we have $F_+<F<F_-$. Let us improve those
estimates.

\begin{lemma}\label{l1c}
If $x<0$ then $F_-(x)+\frac1{ax}<F(x)<F_-(x)$; if $x>0$ then
$F_+(x)<F(x)<F_+(x)+\frac1{ax}$.
\end{lemma}

\begin{proof}
If $x<0$ then
\[
F_-(x)-F(x)=1-\frac1{e^{ax}+1}=\frac{e^{ax}}{e^{ax}+1}<e^{ax}=
\frac1{e^{-ax}}<\frac1{-ax}.
\]
If $x>0$ then
\[
F(x)-F_+(x)=\frac1{e^{ax}+1}<\frac1{e^{ax}}<\frac1{ax}.
\]
\end{proof}

Now we fix $b_0\in\left(\frac{1-\sigma}{2-\sigma},\frac12\right)$ and
assume that $b\in[b_0,1-b_0]$. Set $\phi(b)=2b-\sigma b-1+\sigma$. Note
that $\phi(1-b)=1-2b+\sigma b$. Since $b>b_0$ and $1-b>b_0$, we have
\begin{equation}\label{e1}
\phi(b)\ge\phi(b_0)>0\ \ \ \textrm{and}\ \ \ \phi(1-b)\ge\phi(b_0)>0.
\end{equation}

Set
\[
x_-=-\frac{\phi(b)}{\sigma(2-\sigma)}\ \ \ \textrm{and}\ \ \ x_+=
\frac{\phi(1-b)}{\sigma(2-\sigma)}.
\]
Observe that $x_-<0<x_+$.

\begin{lemma}\label{l2}
We have $F_-(x_-)=x_+$ and $F_+(x_+)=x_-$.
\end{lemma}

\begin{proof}
We have
\[\begin{split}
F_-(x_-)&=(1-\sigma) \frac{-2b+\sigma b+1-\sigma}{\sigma(2-\sigma)}+(1-b)\\ &=
\frac{1-2\sigma+\sigma^2-2b+3\sigma b-\sigma^2 b+2\sigma-\sigma^2-2b\sigma+\sigma^2 b}
{\sigma(2-\sigma)}\\ &=\frac{1-2b+b\sigma}{\sigma(2-\sigma)}=x_+,
\end{split}\]
and
\[\begin{split}
F_+(x_+)=(1-\sigma) \frac{1-2b+\sigma b}{\sigma(2-\sigma)}-b&=
\frac{1-2b+\sigma b-\sigma+2\sigma b-\sigma^2 b-2b\sigma+b\sigma^2}{\sigma(2-\sigma)}\\ &=
\frac{1-2b+\sigma b-\sigma}{\sigma(2-\sigma)}=x_-.
\end{split}\]
\end{proof}

\begin{lemma}\label{l3}
For any $\eps>0$ there exists $\alpha(\eps)$ such that if
$a\ge\alpha(\eps)$ and $\left|x\right|\ge\eps$ then $0\le
F'(x)<1-\sigma$.
\end{lemma}

\begin{proof}
We have
\[
F'(x)=1-\sigma-a\frac{e^{ax}}{(e^{ax}+1)^2},
\]
so $F'(x)<1-\sigma$. Since $\frac{e^{ax}}{(e^{ax}+1)^2}<e^{ax}$, we get
$F'(x)>1-\sigma-ae^{ax}$. Similarly, $\frac{e^{ax}}{(e^{ax}+1)^2}<e^{-ax}$, so
$F'(x)>1-\sigma-ae^{-ax}$. Therefore, $F'(x)>1-\sigma-ae^{-a|x|}$. If
$\left|x\right|\ge\eps$, then we get
$F'(x)>1-\sigma-ae^{-\eps a}$. Since $\lim_{a\to\infty}ae^{-\eps a}=0$,
the lemma follows.
\end{proof}

Set $K=\frac{\phi(b_0)}{2\sigma(2-\sigma)}$ and consider intervals
$I_-=[x_--K,x_-+K]$ and $I_+=[x_+-K,x_++K]$.

\begin{lemma}\label{l4}
There is $a_0$ (depending on $b_0$) such that if $a\ge a_0$ then
$0\le F'(x)<1-\sigma$ for all $x\in I_-\cup I_+$.
\end{lemma}

\begin{proof}
We have $\left|x_-\right|=\frac{\phi(b)}{\sigma(2-\sigma)}$, so
by~\eqref{e1}, for every $x\in I_-$ we get
\[
\left|x\right|\ge\frac{\phi(b)}{\sigma(2-\sigma)}-K\ge
K.
\]
Similarly, $\left|x_+\right|=\frac{\phi(1-b)}{\sigma(2-\sigma)}$, so
for every $x\in I_+$ we get
\[
\left|x\right|\ge\frac{\phi(1-b)}{\sigma(2-\sigma)}-K\ge
K.
\]

So
$\left|x\right|\ge K$ for all
$x\in I_-\cup I_+$.
Set $a_0= \alpha(K)$.
Then, by Lemma~\ref{l3}, if $a\ge a_0$ then $0\le F'(x)<1-\sigma$ for
all $x\in I_-\cup I_+$.
\end{proof}

\begin{lemma}\label{l5}
%If $a\ge a_0$ then 

$F(I_-)\subset I_+$ and $F(I_+)\subset I_-$ for sufficiently large $a$.
\end{lemma}

\begin{proof}
By~\eqref{e1} and Lemmas~\ref{l1c},~\ref{l2} and~\ref{l4}, we have
\[
F(I_-)\subset[F(x_-)-(1-\sigma) K,F(x_-)+(1-\sigma) K]\subset
\left[x_++\frac{1}{ax_-}-(1-\sigma) K,x_++(1-\sigma) K\right].
\]
We have
\[
\frac{1}{ax_-}-(1-\sigma) K=-\frac{\sigma(2-\sigma)}{a\phi(b)}-(1-\sigma)
K\ge-K
\]
for $a\ge \frac{2-\sigma}{K\phi(b)}$, so
\[
F(I_-)\subset[x_+-K,x_++(1-\sigma) K]\subset I_+.
\]

Similarly,
\[
F(I_+)\subset[F(x_+)-(1-\sigma) K,F(x_+)+(1-\sigma) K]\subset
\left[x_--(1-\sigma) K,x_-+\frac{1}{ax_+}+(1-\sigma) K\right].
\]
We have
\[
\frac{1}{ax_+}+(1-\sigma) K\le K,
\]
for $a \ge \frac{2-\sigma}{K\phi(1-b)}$,
so
\[
F(I_+)\subset[x_--(1-\sigma) K,x_-+K]\subset I_-.
\]
Thus, for \[a\ge a_1=\max\left\{a_0,\frac{2-\sigma}{K\phi(b)},\frac{2-\sigma}{K\phi(1-b)}\right\}\]
the lemma follows.
\end{proof}

\begin{theorem}\label{t1}
For a given $\sigma\in(0,1)$ and
$b_0\in\left(\frac{1-\sigma}{2-\sigma},\frac12\right)$, there exists $a_0$
such that if $a\ge a_0$ and $b\in[b_0,1-b_0]$ then $F$ has an
attracting periodic orbit of period $2$ which attracts all points from $(0,1)$ except countably many which fall into $\ol{y}$.
\end{theorem}

\begin{proof}
Let $a_1$ be the constant from Lemma~\ref{l5}. We may additionally
assume that $a_1>4(1-\sigma)$. Then for $a\ge a_1$ we have
$F'(0)=1-\sigma-\frac{a}4<0$, so by Lemma~\ref{l4} the interval $I_-$
lies to the left of 0, and $I_+$ to the right of 0. Therefore,
$I_-\cap I_+=\emptyset$. Now the claim of the theorem follows from
Lemmas~\ref{l4} and~\ref{l5}.

Now we can describe the dynamics of $F$ (and therefore of $f_{ab\sigma}$) %for $a,b$ as in Theorem~\ref{t1}. 
Let $P=\{z_-,z_+\}$, where $z_-<0<z_+$,
be the periodic orbit found there. From the formula for $F'$ it
follows that $F$ has two critical points, $c_-<0$ and $c_+>0$. From
Lemmas~\ref{l4} and~\ref{l5} it follows that
\[
F(c_+)<z_-<c_-<0<c_+<z_+<F(c_-).
\]
In particular, the interval $J=[F(c_+),F(c_-)]$ is invariant.
Moreover, the trajectories of both critical points are attracted to
$P$, so by Lemma~\ref{ls}, there are no attracting or neutral periodic
points except $z_-$ and $z_+$.

Consider intervals $J_-=[F(c_+),c_-]$ and $J_+=[c_+,F(c_-)]$. We have
$0\le F'<1-\sigma$ on $J_-\cup J_+$, so $F(J_-)\subset J_+$ and
$F(J_+)\subset J_-$. Therefore the trajectories of all points from
$J_-\cup J_+$ converge to $P$.

Our map is decreasing on $J_0=[c_-,c_+]$, and has there a fixed point
$z_0$. Since there are no attracting or neutral periodic points in
$J_0$, trajectories of all points from $J_0$ (except $z_0$) are
repelled from $z_0$ and eventually enter $J_-\cup J_+$. Then they are
attracted to $P$. Similar argument shows that trajectories of all
points from $\R\setminus J$ eventually enter $J$, and then they
either hit $z_0$ and stay there, or are attracted by $P$.

\end{proof}

Of course, the same theorem holds if we replace $F$ by $f_{ab\sigma}$, thus we obtain Theorem \ref{thm2period}.
\end{proof}
%
%{\color{blue} To show that there exists and attracting periodic orbit of period 2 we find intervals $I_-$, $I_+$ such that $F(I_-)\subset I_+$ and $F(I_+)\subset I_-$, $I_-\cap I_+=\emptyset$. Then ...}

%{\color{blue} Interpretation of Theorem \ref{2period}}

%For given $\sigma$ and costs close enough with agents using EWA algorithm with sufficiently large intensity of choice system will converge to periodic orbit of period 2. This behavior excludes possibility of {\color{red} any behavior resulting close to perturbed equilibrium when intensity of choice is large enough.} Nevertheless, the long-term behavior is simple.

\subsubsection{Long-term behavior outside $ \left(\frac{1-\sigma}{2-\sigma},\frac{1}{2-\sigma}\right)$}
Now we look at the case $b \in \left(0, \frac{1-\sigma}{2-\sigma} \right) \cup \left( \frac{1}{2-\sigma}, 1\right)$. 
Parameter $b$ is the characteristic of our game -- it tells us how different are the costs of resources. Taking $b$ far enough from $1/2$ implies that costs of resources are distinguishable for agents who are discounting/forgetting past costs with factor $\sigma$. We show that in such case the chaotic behavior emerges.%\footnote{{\color{red} Can we state the following: approaching NE at the cost of destabilizing the system???}}
%Once more we will use two linear maps,

\begin{theorem} \label{chaos}
For a fixed $\sigma>0$ and $b$, if either $b\le 1/2$ and
$b<\frac{1-\sigma}{2-\sigma}$, or $b\ge 1/2$ and $b>\frac{1}{2-\sigma}$, then there
exists $a_1>0$ such that if $a>a_1$ then $f_{ab\sigma}$ is Li-Yorke chaotic and  $h (f_{ab\sigma})\geq \log\frac{1+\sqrt{5}}{2}$.%\footnote{{\color{red} what with $(1-\sigma)/(2-\sigma)$ and $1/(2-\sigma)$?}}
\end{theorem}
%The proof of existence of Li-Yorke chaos has several steps. First, we show existence of a point such that its first and third iterates by the map $F$ are located on opposite sides of the point. This property implies existence of a periodic orbit of period 3. From here we apply Sharkovsky Theorem to obtain the existence of periodic orbits of all periods and Li-Yorke chaos for the map $F$. Finally, as $F$ and $\fabs$ are topologically conjugate we get Li-Yorke chaos for the map $\fabs$.

\begin{proof}%[Proof of Theorem \ref{chaos}]
First, we show auxiliary lemmas.
\begin{lemma}\label{l2c}
If $a$ is sufficiently large, then the map $F$ has two critical
points, $c_-<0$ and $c_+>0$, independent of $b$. For $\sigma$ fixed,
those points go to $0$ as $a$ goes to infinity.
\end{lemma}

\begin{proof}
We have
\[
F'(x)=1-\sigma-\frac{ae^{ax}}{(e^{ax}+1)^2}.
\]
Thus, if $t=e^{ax}$, then the equation for the zeros of $F'$ becomes
\[
t^2+\left(2-\frac{a}{1-\sigma}\right)t+1=0.
\]
If $a$ is sufficiently large, then this equation has two roots, both
positive, and one less than 1 and the other one larger than 1.

For a given $\eps>0$, as $a$ goes to infinity, $F'$ converges
uniformly to $1-\sigma$ on $\R\setminus(-\eps,\eps)$. Thus, if $a$ is
sufficiently large, both critical points have to be in $(-\eps,\eps)$.
\end{proof}

\begin{lemma}\label{l3c}
Fix $\sigma$ and $b$. Assume that $b\le 1/2$ and
$b<\frac{1-\sigma}{2-\sigma}$. Then there is $a_0$ such that if $a>a_0$ then
there is a point $x_0\in(c_-,c_+)$ such that $F(x_0)=c_-$ and
$F^3(x_0)>c_+$.
\end{lemma}

\begin{proof}
By Lemmas~\ref{l1c} and~\ref{l2c}, as $a$ goes to infinity, then
$F(c_-)$ goes to $1-b$ and $F(c_+)$ goes to $-b$. We have $F(0)\ge
0>c_-$ and $F(c_+)<c_-$, so there is a point $x_0\in(c_-,c_+)$ such
that $F(x_0)=c_-$. For $a$ sufficiently large, we get $F^3(x_0)$
arbitrarily close to $(1-\sigma)(1-b)-b$, which is positive, while $c_+$
is arbitrarily close to 0. Thus, $F^3(x_0)>c_+$.
\end{proof}

From this we get an immediate corollary.

\begin{theorem}\label{t1c}
For a fixed $\sigma>0$ and $b$, if either $b\le 1/2$ and
$b<\frac{1-\sigma}{2-\sigma}$, or $b\ge 1/2$ and $b>\frac{1}{2-\sigma}$, then there
exists $a_0$ such that if $a>a_0$ then $F$ has a periodic point of
period $3$.
\end{theorem}

\begin{proof}
In the first case this follows from Lemma~\ref{l3c}. In the second
case, use the first one and the identity $F_{1-b}(-x)=-F_b(x)$.
\end{proof}

Existence of a periodic point of period 3 implies Li-Yorke chaos of the system \cite{SKSF} for $F$ and thus for $f_{ab\sigma}$.
 Finally, for any interval map if it has periodic point of period 3, then the entropy of this maps is at least $\log \frac{1+\sqrt{5}}{2}$, \cite{BGMY}.  
This completes the proof of Theorem \ref{chaos}.
\end{proof}

\begin{corollary} \label{corchaos}
For fixed $\sigma\in [0,1)$ and $b \in \left(0, \frac{1-\sigma}{2-\sigma} \right) \cup \left( \frac{1}{2-\sigma}, 1\right)$ the system becomes chaotic for sufficiently large values of the intensity of choice.
\end{corollary}

\begin{comment}
\begin{figure}
\centering
\includegraphics[width=0.8\textwidth]{cp2} %0.68
\begin{minipage}{0.9\textwidth}
\caption{\small Behavior of the system for large values of intensity of choice $a$. As long as $b\in \left(\frac{1-\sigma}{2-\sigma},\frac{1}{2-\sigma}\right)$, almost all trajectories are attracted to the periodic orbit of period 2. Outside this interval we observe chaotic behavior.}
\end{minipage}
\label{fig: cp2}
\end{figure}
\end{comment}

%For a fixed $\sigma\in [0,1)$ our results are the following (see Figure \ref{fig: periodic})
%\begin{enumerate}
%\item If $b\in I_{\sigma}=\left(\frac{1-\sigma}{2-\sigma},\frac{1}{2-\sigma}\right)$, then for sufficiently large $a$ almost every $f_{ab\sigma}$-trajectory converges to the attracting orbit of period 2 (Theorem \ref{2period} and Corollary \ref{cor2period})
%\item For  $b\in (0,1)\backslash I_{\sigma}$ for sufficiently large $a$ the map $f_{ab\sigma}$ has periodic orbits of every period, is Li-Yorke chaotic and has {\color{red} positive topological entropy}. (Theorem \ref{chaos}).%\footnote{{\color{red} Econ comment on definitions and interpretations of chaos.}}
%\end{enumerate}

Theorem \ref{chaos} shows us that when the difference in costs of resources is substantial if agents choose their strategies with sufficiently large intensity of choice $a$, then the system will inevitably become chaotic. In such case any long-term behavior will become extremely complex.
% to study.
%Thus, there is range of values of $b$ (and thus, different costs) descending with $\sigma$ in which the growth of the intensity of choice will lead to the chaotic behavior of the system. 
On the other hand, when the cost of resources are similar enough, memory loss makes those costs indistinguishable from the perspective of an agent. In such case when the intensity of choice is large the agents follow an attracting periodic orbit of period 2 (Theorem \ref{thm2period}).
%\footnote{Are we using intensity of choice or demand???}
This is a crucial differentiation of the long-term behavior of the system: existence of periodic orbit of period 2 which attracts almost all trajectories implies that although the system does not stabilize, it remains relatively predictable --- no matter the initial state of the system, it will converge to period 2 orbit, thus after some time, every even number of iterations of the map will place it close to its previous position. When the system becomes chaotic we land in an unpredictable regime
%at the antipode of this behavior 
with periodic orbits of different periods, dependence on initial conditions and complicated dynamics.
%In which case we end up decides an interplay between $b$ and $\sigma$.

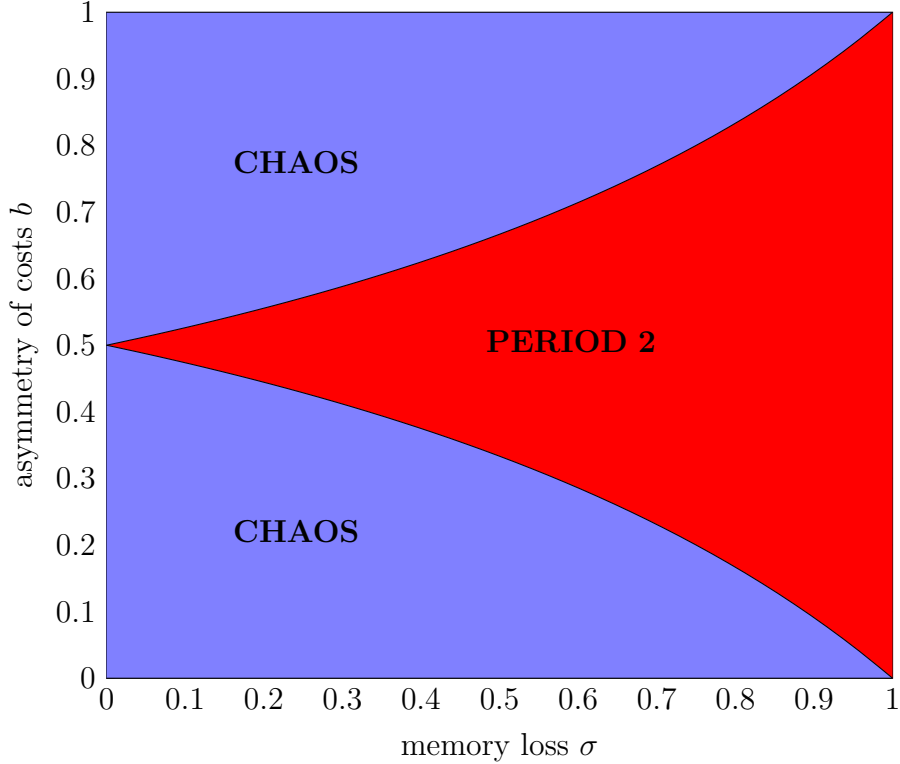
\begin{figure}
\centering
\begin{tikzpicture}[domain=0:1,scale=0.7]
  \begin{axis}[thick,smooth,no markers, 
  width=\linewidth,
  line width=0.5,
  grid=major, % Display a grid
  tick label style={font=\normalsize},
  legend style={nodes={scale=0.4, transform shape}},
  label style={font=\normalsize},
  legend image post style={mark=triangle},
  grid style={white},
  xlabel={memory loss $\sigma$},
  xlabel style={yshift=-0.2cm}, 
ylabel={asymmetry of costs $b$},
ylabel style={yshift=0.2cm}, 
   %y tick label style={
   % /pgf/number format/.cd,
    %fixed,
    %fixed zerofill,
   % precision=4
 %},
legend style={at={(1,1)}, anchor=north east,  draw=none, fill=none},
ymin = 0,
ymax = 1,
xmin=0,
xmax=1
  ]
           \addplot+[name path=A,black] {(1-x)/(2-x)};
        \addplot+[name path=B,black] {1/(2-x)};
        \addplot+[name path=C,black] {0};
        \addplot+[name path=D,black] {1};
        \addplot[red] fill between[of=A and B];
        \addplot[blue!50] fill between[of=A and C];
        \addplot[blue!50] fill between[of=B and D];
    
          \draw[xshift=7.0cm,yshift=6.4cm] node [right,text width=6cm]{{\bf PERIOD 2}};
  \draw[xshift=2.2cm,yshift=9.8cm] node [right,text width=6cm]{{\bf CHAOS}};
  \draw[xshift=2.2cm,yshift=2.8cm] node [right,text width=6cm]{{\bf CHAOS}};

   % \addplot[blue] coordinates {(100,0.9245) (200,0.9247) (300,0.9243) (400,0.9244) (500,0.9246)(750,0.9253) (1000,0.9248)};  \addlegendentry{Test}

%     \addplot[no marks,dotted,red] coordinates
%      {(100,0.6949) (200,0.7114) (300,0.7130) (400,0.7134) (500,0.7135 )(750,0.7134) (1000,0.7135)}; 
%      \addlegendentry{Error de entrenamiento}

   % \addplot[brown] coordinates {(100,0.8814) (200,0.8836) (300,0.8847) (400,0.8855) (500,0.8857)(750,0.8862) (1000,0.8863)};
    % \addlegendentry{Test} 

%      \addplot[black,dotted] coordinates
%     {(100,0.7285) (200,0.7564) (300,0.7671) (400,0.7721) (500,0.7745)(750,0.7766) (1000,0.7770)};
%     \addlegendentry{Error de entrenamiento}
  \end{axis}
\end{tikzpicture}
\begin{minipage}{0.9\textwidth}
\caption{\small Behaviors of the system for large values of the intensity of choice $a$. As long as $b\in \left(\frac{1-\sigma}{2-\sigma},\frac{1}{2-\sigma}\right)$, almost all trajectories are attracted to the periodic orbit of period 2. Outside this interval we observe chaotic behavior.}
\label{fig: cp2}
\end{minipage}
\end{figure}

Corollaries \ref{cor2period} and \ref{corchaos}  determine the sets of parameters $(\sigma,b)$  in which the long-term behavior for large intensity of choice is diametrically different (see Figure \ref{fig: cp2}). When $\sigma=0$ the interval $\left(\frac{1-\sigma}{2-\sigma},\frac{1}{2-\sigma}\right)$ shrinks to $\{1/2\}$ and the system will be chaotic if only cost functions of different paths are different. When $\sigma$ increases the interval where we observe attraction to the orbit of period 2 expands. As $\sigma$ tends to 1 chaotic behavior vanishes and in the instability region
%(which by Proposition \ref{QREstability} is growing larger as $\sigma$ increases)
almost all trajectories (except countably many) converge to the attracting periodic orbit of period 2. This result agrees with what we know for these special cases (see Proposition \ref{dynsigma0} here and Theorem 3.7 with Corollary 3.10 from \cite{CFMP2019}).%\footnote{{\color{red} Can we say something more about this attracting periodic orbit?}}
%For sufficiently large $a$ system becomes unstable.
We finally note that the phase transition at $(\sigma,\frac{1-\sigma}{2-\sigma})$ and $(\sigma,\frac{1}{2-\sigma}) $ implies that close to these values a small change of costs of resources (change of $b$) as well as small change in memory of the agents (change of $\sigma$) can push the system from simple periodic behavior to the complex chaotic one (or in the opposite direction).%\footnote{{\color{red}expand}}
%\begin{remark} $\sigma>1$ \end{remark}

\section{Conclusions}

In this paper we show that chaotic behavior can be observed in a large class of EWA dynamics for simple two-strategy nonatomic congestion game. We derive this class of dynamics from Galla and Farmer \cite{GallaFarmer_PNAS2013}. We show that in such game an increase in the intensity of choice will inevitably result in loosing stability of the system.
 %the perturbed equilibrium will inevitably loose stability, which will result in the instability of the system. 
 Moreover, the interplay between asymmetry of costs and memory loss will give qualitatively different behaviors for large values of the intensity of choice. For $\sigma=0$, that is when all previous costs are equally important, the system will become chaotic only if costs of resources are different. When $\sigma$ increases (memory loss/discount factor increases) the range of values of the parameter of asymmetry of costs $b$, for which the trajectories of almost all points will be attracted by periodic orbit of period 2, will grow, eventually for $\sigma=1$ attaining the whole unit interval $(0,1)$. This behavior gives two completely different regimes. The system where all trajectories are attracted to the periodic orbit of period 2 is predictable and the dynamics is simple, while chaotic regime is unpredictable resulting in complex dynamics.
%Thus, an increase in memory loss can prevent chaotic behavior.

%Putting these results in the wider context, they suggest 
Our results show that while potential/congestion games are traditionally  viewed as one of the most predictable classes of games in terms of their dynamics, their detailed picture is much more complicated.
These results are in line with numerous recent findings \cite{BCFKMP21,CP2019,CFMP2019,Thip18,mertikopoulos2017cycles,Farmer2019,GallaFarmer_ScientificReport18}, suggesting that complex and non-equilibrating behavior of agents employing learning rules %\footnote{Such as Multiplicative Weights Update (MWU), Expected Weights Attraction (EWA), Follow the Regularized Leader (FTRL) or ficticious play.} 
widely applied in economics %and other sciences 
seems to be common rather than exceptional. 

In addition, we show that memory loss can prevent chaos in two-strategy congestion game with homogeneous population of agents. But what will happen in heterogeneous case? And what if agents have more strategies/resources available? Evidently, the system will be more complicated. Nevertheless, in the full memory case one can observe the emergence of chaotic behavior for $b\neq 1/2$ as a consequence of the increase of the intensity of choice, both in heterogeneous case and for many strategies \cite{CFMP2019}. We leave the answer to the memory loss case for future work. Moreover, one may ask if these results are algorithm specific. Results on more general classes of dynamics \cite{BCFKMP21,mertikopoulos2018riemannian} suggest that our result can be generalized to larger class of dynamics like (discounted) FTRL dynamics. We also leave the answer to this question for the future work.

Lastly, Pangallo et al. \cite{Farmer2019} showed that best reply cycles --- basic topological structures in games --- predict nonconvergence of six well-known learning algorithms that are used in biology or are supported by experiments with human players. Best reply cycles are dominant in complicated and competitive games, indicating that in these cases equilibrium is typically an unrealistic assumption, and one must explicitly model the learning dynamics. These examples of complex and chaotic behavior strongly suggests that chaotic, non-equilibrium results can be further generalized to other  %variants of zero-sum
 games.
 
 %As Mertikopoulos and Sandholm result on Riemannian game dynamics \cite{MS2018} suggests that for large class of penalty dynamics the difference in behavior depends mainly on geometric properties, not topological ones. Therefore, as chaos is a topological property, the chaotic behavior detected for a special case of MWU should also occur for FTRL dynamics

\section*{Acknowledgements}

Georgios Piliouras acknowledge AcRF Tier 2 grant 2016-T2-1-170, grant PIE-SGP-AI-2018-01, NRF2019-NRF- ANR095 ALIAS grant and NRF 2018 Fellowship NRF-NRFF2018-07.
Research of Micha{\l} Misiurewicz was partially supported by grant number 426602 from the Simons Foundation.
Jakub Bielawski and Fryderyk Falniowski acknowledge support from a subsidy granted to Cracow University of Economics and COST Action CA16228 ``European Network for Game Theory''. Thiparat Chotibut was supported by Thailand Science Research and Innovation Fund Chulalongkorn University [CU\_FRB65\_ind (5)\_110\_23\_40].

\bibliographystyle{abbrv} 
\bibliography{sigma-jet-arxiv.bib}
%\bibliography{ms2}

\end{document}